\documentclass[10pt, conference, letterpaper]{IEEEtran}

\usepackage[cmex10]{amsmath}
\usepackage{amssymb}
\usepackage{cite}
\usepackage{amsmath}
\usepackage{amssymb}
\usepackage{wasysym}
\usepackage{booktabs}
\usepackage{paralist}
\usepackage{balance}
\usepackage{subfigure}
\usepackage{multicol}
\usepackage{graphicx}
\usepackage{epsfig}
\usepackage{epstopdf}
\usepackage{algorithmic}
\usepackage[ruled,linesnumbered]{algorithm2e}
\usepackage{bm}
\usepackage{amsmath, amssymb,amsthm}
\usepackage{upgreek}
\usepackage{makecell}
\usepackage{caption}
\usepackage{subfigure}
\usepackage{color}
\theoremstyle{definition} 
\theoremstyle{definition} 
\theoremstyle{plain} 
\theoremstyle{plain} \newtheorem{proposition}{Proposition}

\begin{document}

%
\title{Energy and Delay Optimization for Cache-Enabled Dense Small Cell Networks}
%
%

%
 \author{\IEEEauthorblockN{Hao Wu, Hancheng Lu}
 \IEEEauthorblockA{~The Information Network Lab of EEIS Department USTC, Hefei, China, 230027 \\
 hwu2014@mail.ustc.edu.cn, hclu@ustc.edu.cn}}
\maketitle


\begin{abstract}
Caching popular files in small base stations (SBSs) has been proved to be an effective way to reduce bandwidth pressure on the backhaul links of dense small cell networks (DSCNs). Many existing studies on cache-enabled DSCNs attempt to improve user experience by optimizing end-to-end file delivery delay. However, under practical scenarios where files (e.g., video files) have diverse quality of service requirements, energy consumption at SBSs should also be concerned from the network perspective. In this paper,we attempt to optimize these two critical metrics in cache-enabled DSCNs. Firstly, we formulate the energy-delay optimization problem as a Mixed Integer Programming (MIP) problem, where file placement, user association and power control are jointly considered. To model the tradeoff relationship between energy consumption and end-to-end file delivery delay, a utility function linearly combining these two metrics is used as an objective function of the optimization problem. Then, we solve the problem in two stages, i.e. caching stage and delivery stage, based on the observation that caching is performed during off-peak time. At the caching stage, a local popular file placement policy is proposed by estimating user preference at each SBS. At the delivery stage, with given caching status at SBSs, the MIP problem is further decomposed by Benders' decomposition method. An efficient algorithm is proposed to approach the optimal association and power solution by iteratively shrinking the gap of the upper and lower bounds. Finally, extension simulations are performed to validate our analytical and algorithmic work. The results demonstrate that the proposed algorithms can achieve the optimal tradeoff between energy consumption and end-to-end file delivery delay.
\end{abstract}


%
\IEEEpeerreviewmaketitle

\section{Introduction}
%
%
%
%

To cope with the explosive mobile traffic growth, dense small cell networks (DSCNs) are expected to be deployed in fifth generation (5G) cellular networks. In DSCNs, small base stations (SBSs) are usually connected to the core network via low-capacity backhaul links due to physical and cost-related limitations \cite{1u}\cite{3u}. That means the backhaul is prone to be the system bottleneck. Moreover, the backhaul problem becomes more serious as the SBS deployment density increases. Recently, enabling cache in DSCNs have been considered as a promising way to handle the backhaul problem \cite{2u,backhaulcache,6u,3u}. Statistical report has shown that a few popular files requested by many users should account for most of backhaul traffic load \cite{2u}. Based on this fact, popular files can be proactively cached at SBSs, and delivered to users when requested, without consuming backhaul bandwidth. The effect of caching on the backhaul is determined by file reuse, i.e., the number of users requesting the same file. If there is enough file reuse, caching can replace backhaul communication \cite{3u}.

In cache-enabled DSCNs, user experience is also improved due to the reduction of end-to-end file delivery delay\cite{1u}\cite{6u,cachingdelay2,5u,cachingdelay1}. When a user requests a file cached in the local SBS, the file is delivered by that SBS instead of the faraway Internet file server. In this case, end-to-end file delivery delay is significantly reduced. We can also see that minimizing end-to-end file delivery delay is equivalent to maximizing the cache hit ratio. Many existing studies attempt to improve the cache hit ratio by optimizing file placement in cache of SBSs \cite{1u}\cite{3u}\cite{5u}. However, the file placement optimization problem is non-trivial, which is coupled with the file popularity distribution (i.e., the probability that a file is requested by users) and user association strategy. When the file popularity distribution is known at each SBS, the file placement optimization problem can be converted to a well-known knapsack problem. Learning-based algorithm are proposed to obtain the file popularity profile and cache the best files at SBSs when the file popularity distribution is not known \cite{1u}. When users can associate with multiple SBSs, a distributed caching optimization problem is formulated based on a connectivity bipartite graph model and approximation algorithms that lie within a constant factor of the theoretical optimum are proposed \cite{3u}. In cache-enabled DSCNs with mobile users, the file replacement problem is optimized with recommendation via Q-learning \cite{5u}.

Unlike end-to-end file delivery delay, energy consumption, which is widely concerned in 5G cellular networks\cite{5G1}, has not been well studied in cache-enabled DSCNs. There only exist a few studies on this issue. In \cite{6u}, the impact of various factors (backhaul capacity, content popularity, cache capacity, etc.) on downlink energy efficiency (EE) is analyzed. It also validates that caching in DSCNs can achieve more EE gain compared with caching in conventional cellular networks. However, in this work, the user association strategy and quality of service (QoS) requirements from files are not considered. On the one hand, in DSCNs, multiple SBSs are available for a user that locates at the edge of a small cell. It means that the user have multiple association choices. In this case, the user association strategy has a significant impact on energy consumption\cite{EE1}\cite{EE2}. On the other hand, in practice, files have diverse QoS requirements. As we known, video traffic plays a major part in current mobile traffic, and is predicted to contribute over 80\% of total mobile traffic in 2020 \cite{7u}. For online video delivery, the QoS requirement of each video file is usually expressed in forms of rate \cite{filerequirment}. Consequently, different transmission power levels are configured at SBSs for video delivery to support required rates under various channel conditions.

\begin{figure}[t]
  \centering
  \includegraphics[width=0.43\textwidth]{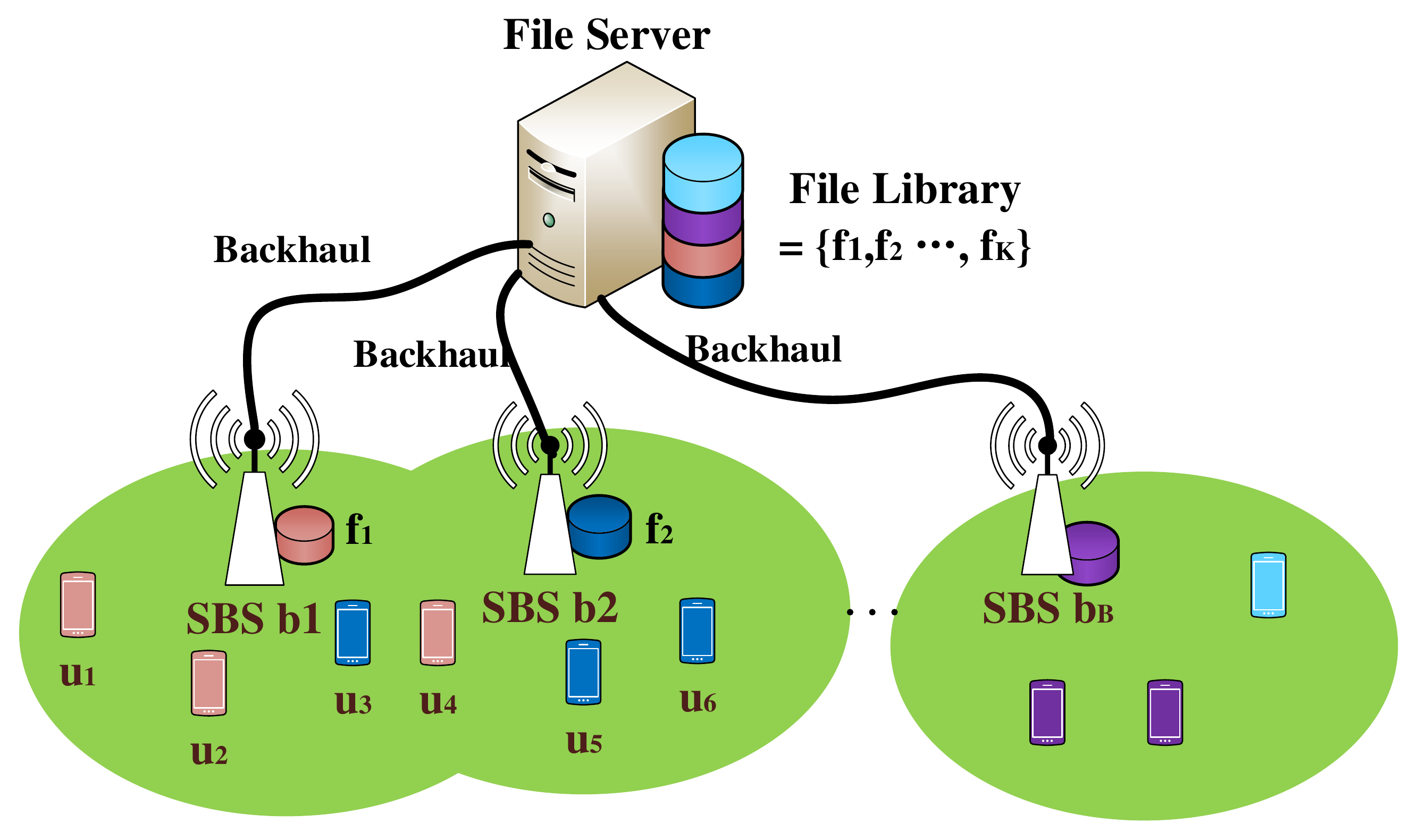}
  \caption{Cache-enabled DSCNs.}
  \label{example}
\vspace{-3ex}
\end{figure}
Both end-to-end file delivery delay and energy consumption are critical metrics in cache-enabled DSCNs \cite{5G}. From the perspective of user experience, lower end-to-end file delivery delay are preferred, while less energy consumption is preferred from the network perspective. Unfortunately, there is a contradiction between these two metrics. Considering two SBSs, i.e., b1 and b2, as well as their small cells described in Fig. \ref{example}, we assume that files have different QoS requirements in terms of rate. $u_{1}$, $u_{2}$ and $u_{4}$ request $f_{1}$, and others request $f_{2}$. To minimize end-to-end file delivery delay, the optimal user association strategy is: $u_{1}$, $u_{2}$ and $u_{4}$ associate with $b_{1}$, others associate with $b_{2}$. In this case, users can download files from SBSs, without backhual delay. However, a different user association strategy should be applied to achieve optimal energy consumption, which is: $u_{1}$, $u_{2}$ and $u_{3}$ associate with $b_{1}$, others associate with $b_{2}$. In this case, users associate with their nearest SBSs. Under practical scenarios, the tradeoff relationship between energy consumption and end-to-end file delivery delay is more complicated. In this paper, for the first time, we study suck kind of tradeoff in cache-enabled DSCNs. The main contributions are described as follows.

1) In cache-enabled DSCNs, we analyze end-to-end file delivery delay and energy consumption. Based on our analysis, we formulate the energy-delay tradeoff problem as a mixed integer programming (MIP) problem, where file placement, user association and power control are jointly considered.

2) To alleviate traffic pressure on the backhaul, file placement is performed during off-peak time. Based on this fact, we propose a local popular file placement policy at each SBS. In the proposed policy, the optimal file placement problem is converted to a knapsack problem and solved by an efficient greedy algorithm.

3) With the proposed file placement policy, the energy-delay tradeoff problem is reduced to a mixed integer linear programming (MILP) problem and is further decomposed with Benders¡¯ decomposition method. Then, an efficient algorithm is proposed to approach the optimal association and power solution by iteratively shrinking the gap of the upper and lower bounds.

Extension simulations are carried out to validate our theoretical and algorithmic work. The results demonstrate that the proposed algorithms have fast convergence speed and can achieve the desired tradeoff between energy consumption and end-to-end file delivery delay.

The rest of the paper is organized as follows. Section II gives an overview of the system model. The energy-delay tradeoff problem is formulated in cache-enabled DSCNs in Section III. In Section IV, a local popular file placement policy is proposed. Then, the energy-delay tradeoff problem is decomposed and solved based on Benders' decomposition method. Performance evaluation is presented in Section V. Finally, we conclude the paper in Section VI.

%


\section{system model}

\begin{table}
  \centering
  \caption{NOTATIONS}
  \label{summary of the notation used in this paper}
  \begin{tabular}{|c|c|}
    \hline
    \textbf{Symbol} & \textbf{Description}\\\hline
    $b_j$, $u_i$, $f_k$ &  SBS, user and file indexed by $j$, $i$, $k$ respectively \\\hline
    $p_j, e_j$ &  Transmission power and energy consumption at $b_j$ \\\hline
    $C_j$ &  Cache capacity of $b_j$ \\\hline
    $\tau_{k}$ &  Wireless transmission delay of $f_k$ \\\hline
    $\rho_{ik}$ & $u_i$'s preference for $f_k$\\\hline
    $\psi_{jk}$ & Popularity of $f_k$ at $b_j$\\\hline
    $d_{ij}^{k}$ & \makecell{File delivery delay for $f_k$ when $u_i$ associates with $b_j$} \\
  \hline
    $\theta_{ik}\in \{0,1\}$ & If $u_i$ requests $f_k$, $\theta_{ik}=1$.\\\hline
    $x_{ij}\in \{0,1\}$ & If $u_i$ associates with $b_j$, $x_{ij}=1$.\\\hline
    $y_{ik}\in \{0,1\}$ & If $f_k$ is in cache of $b_j$, $y_{ik}=1$.\\
      \hline

   \end{tabular}\label{notation}
   \vspace{-3ex}
\end{table}
Consider a downlink DSCN consisting of $B$ SBSs (i.e., femto base stations or pico base stations) indexed by a set $\mathcal{B}$ =\{1, 2, ..., $B$\}, as shown in Fig. 1. All SBSs are cache-enabled and the cache capacity of SBS $b_j$ is denoted by $C_j$ ($j\in \mathcal{B}$). $U$ users are randomly deployed in the coverage of DSCNs. Let $\mathcal{U}$ denote the user index set and $\mathcal{U}$ = \{1, 2, ..., $U$\}. Assume each user $u_i$ ($i\in \mathcal{U}$) can only associate with one SBS.

The requested files are indexed by a set $\mathcal{F}$ = \{1, 2, ..., $F$\}, which are stored as a file library at the file server and cached at SBSs according to file placement policies. For file $f_k$ ($k\in \mathcal{F}$), its size is denoted by $s_k$ and its QoS requirement in terms of rate denoted by $R_k$. It means that the transmission rate of $f_k$ should be not less than $R_k$. SBS $b_j$ employs power $p_j$ to transmit files and satisfy their rate requirements.

Some major notations are summarized in Table \ref{notation}.

\subsection{Interference Model}
Orthogonal frequency division multiplexing (OFDM) is assumed to be used in cache-enabled DSCNs.  In each small cell, resource blocks (i.e., time and frequency) allocated to different users are orthogonal. Therefore, only inter-cell interference from neighbouring cells on the same frequency is considered.

For user $u_i$ associating with SBS $b_j$, inter-cell interference is $\sum\limits_{l\in\mathcal{B},l\neq j} p_{l}g_{il}$, where $p_{l}$ is the transmission power of SBS $b_{l}$ and $g_{il}$ denotes the channel gain between $u_{i}$ and $b_{l}$. Thus, signal-interference-noise-ratio (SINR) at $u_i$ associating with $b_j$ can be expressed as
\begin{equation*}
  \gamma_{ij} = \frac{p_{j}g_{ij}}{\sum\limits_{l\in\mathcal{B},l\neq j} p_{l}g_{il}+\sigma^2},
\end{equation*}
where $\sigma^2$ denotes the noise power level. The downlink data rate (bit/s) of $u_{i}$ is
\begin{equation}\label{shannon}
  r_{ij}=Wlog_{2}(1+\gamma_{ij}),
\end{equation}
where $W$ is the bandwidth allocated to one user. In this paper, we consider practical scenarios where files (e.g., video files) have diverse qualities of service requirements. When $u_i$ requestes file $f_{k}$, $r_{ij}$ must satisfy $r_{ij}  \geq R_{k}$ to guarantee the delivery quality.

\subsection{Local File Popularity Distribution}
In Figure. 1, each SBS is equipped with cache.  Files are placed in cache according to the file popularity distribution. Usually, the global file popularity distribution (e.g. Zipf distribution) is used and according to popularity rankings the each SBS will cache the same files \cite{2u}\cite{6u}. However, such coarse-grained file placement policy ignores the local file popularity characteristic of each small cell, which incurs the waste of cache resource. Local popularity means that different sociological and cultural backgrounds of users at different locations and popularities vary from region to region. Related researches and studies have shown that files or contents have a local popularity characteristic. In \cite{IIB-localpopularity}\cite{IIB-userpreference}, authors studied the file popularity distribution and users¡¯ access patterns in video traffic from a campus network. By analysis of real-world trace data, some significant conclusions are made as follows: 1) global file popularity can not reflect the local file popularity (e.g. a file popularity in a cell site), so the file popularity among different cell sites may be different from each other. 2) users may have strong personal preferences toward their specific file categories.

It is reasonable to assume that during $T$, the popularity distribution of the files and the user preference for the files are fixed. Typical examples include popular news and short videos, which are updated every $1$-$3$ hours.  Besides, another important assumption is made that during the same time scale $T$, the users in the covering area of each cell are fixed. That is to say, the users move slowly within the covering of the cell during $T$. Therefore, for the convenience of this research and considering the slow change of the file popularity, we first intend to consider the fixed file popularity of one $T$.

Based on the above assumptions, to obtain a local file popularity distribution, each SBS first needs to estimate the preference of the central-zone users during off-peak time and then calculates the local file popularity based on the user preference. In a small cell with SBS $b_{j}$, based on users' distance to $b_{j}$, users are divided into central-zone users and edge users. Central-zone users are just around $b_{j}$ and are more likely to associates with $b_j$ than edge users from the perspective of energy consumption. Considering SBS $b_{j}$, the users index set in the central covering area is denoted by $\Phi_{j}=\{i |$ $u_i$ is in the central zone of Cell $b_j\}$. For user preference, we adopt the definition and model of the user preference similar to that in \cite{IIB-userpreference}, where the user preference is modeled by kernel function. Kernel function can efficiently reflect the correlation between the user and the file. Let $\psi_{jk}$ denotes the popularity of file $f_k$ at SBS $b_{j}$, which is the weighted sum of probabilities that central-zone users request the file $f_k$. Then we can get the local file popularity distribution at SBS $b_{j}$:
\begin{equation}\label{filepopular}
  \psi_{jk}=\sum_{i\in \Phi_{j}} p(u_{i})\rho_{ik},~~~~~f_k \in \mathcal{F}
\end{equation}
where $p(u_{i})$ is the probability that user $u_i$ generates a file request. $\rho_{ik}$ is used to stand for $u_{i}$'s preference for file $f_k$. $\psi_{jk}$ denotes the local popularity at $bj$ and also reflects the ratio of the requests for $f_k$ to the total ones in the central-zone of the small cell at any moment.

\section{Problem Formulation}
In this paper, we attempt to optimize end-to-end file delivery delay and energy consumption by joint power control, user association and file placement.  Firstly, we analyze end-to-end file delivery delay and energy consumption of SBSs in cache-enabled DSCNs, respectively. Then, based on our analysis, we formulate the optimization problem as an MILP problem. In this section, end-to-end file delivery delay and system energy consumption in cache-enabled DSCNs are analyzed. Based on our analysis, we formulate the energy-delay optimization problem.


\subsection{Delay and Energy Consumption Analysis}
In cache-enabled DSCNs, each SBS is equipped with cache. If SBS caches the file that a user requests, the end-to-end file delivery delay for this user equals to wireless transmission delay. Otherwise, the end-to-end file has to be delivered by the remote file server, and additional backhaul delay is involved. Let $d_{ij}^{k}$ denote the end-to-end file delivery delay when user $u_i$ associating with SBS $b_j$ requests file $f_k$.
We have
\begin{equation}\label{delay function}
  d_{ij}^{k}=
\begin{cases}
\tau_{ij}^{k},              &\quad~~ y_{jk} = 1,\\
\tau_{ij}^{k}+ w_{j}^{BH},  &\quad~~y_{jk} = 0,
\end{cases}
\end{equation}
where binary variable $y_{jk}$ indicates whether file $f_k$ is cached in $b_j$ or not, and $\tau_{ij}^{k}=\frac{s_{k}}{r_{ij}}$ represents wireless transmission delay of $f_k$ transmitted from $b_j$ to $u_i$. Backhaul delay of $b_j$ is denoted by $w_{j}^{BH}$.  For wired backhaul, backhaul delay of SBSs is related to the average link distance, the average traffic load and the average number of SBSs connecting to a remote file server in Internet core. Hence, backhaul delay $w_{j}^{BH}$ at SBS $b_j$ can be modeled to be an exponentially distributed random variable with a mean value of $D_j$\cite{19-backhauldelay}.


Here the user request model is given and it is assumed that each user request only one file once time. Let $\theta_{ik}=1$ denote whether user $u_i$ requests file $f_k$ or not. $\theta_{ik}=1$ when the $u_i$ requests $f_k$. Otherwise, $\theta_{ik}=0$. And $\sum_{k}\theta_{ik}=1$ makes $u_i$ only request one file once time. Then, we can derive end-to-end file delivery delay for $u_i$ as follows.
\begin{equation}\label{delay}
   d_{i}=\sum_{j\in \mathcal{B}}\sum_{k\in \mathcal{F}} \theta_{ik}x_{ij}(\tau_{ij}^{k}+(1-y_{jk})w_{j}^{BH}),
\end{equation}
where $x_{ij}\in\{0, 1\}$ is a binary variable. If user $u_i$ associates with SBS $b_j$, $x_{ij} = 1$. Otherwise, $x_{ij} = 0$. Then, one of our optimization objectives is to minimize end-to-end file delivery delay of all users:
\begin{equation*}
  \min_{\bm{X},\bm{p},\bm{Y}} \sum_{i\in\mathcal{U}}d_{i}
\end{equation*}

Energy consumption at SBS $b_j$ is expressed as follows:
\begin{equation*}
  e_{j}= p_jT_{j},
\end{equation*}
where file serving time $T_j=\sum_{i\in \mathcal{U}}\sum_{k\in \mathcal{F}}\theta_{ik}x_{ij}\tau_{ij}^{k}$ at SBS $b_j$ denotes the time required to complete transmission of all requested files at $b_j$. 

Then, the other objective is to minimize total transmission energy consumption:
\begin{equation*}
  \min_{\bm{X},\bm{p}} \sum_{i\in\mathcal{U}} e_{j},
\end{equation*}
\subsection{Energy and Delay Optimization}
Compared with optimizing the two objectives separately, jointly optimizing energy and delay belongs to a kind of multi-objective optimization problem. To express such optimization problem, we employ a weighted sum based utility function, which is modeled by the weighted sum of energy cost and delay cost \cite{II-B:WeightedSum}.
We can formulate the energy-delay optimization problem as follows.
\begin{align}\label{tradeoff}
\min_{\bm{X},\bm{p},\bm{Y}}    &\quad\alpha\sum_{j\in \mathcal{B}} e_j+ (1-\alpha)\sum_{i\in \mathcal{U}} d_i\\
\text{s.t.}    &\quad 0<\sum_{i\in \mathcal{U}}p_jx_{ij}\leq P_{j}^{max}\label{constraint:power_constraint},\forall j\in \mathcal{B},\\
        &\quad r_{ij}\geq R_{k}, \forall i\in \mathcal{U}, ~ k\in \mathcal{F} \label{constraint:file datarate},\\
        &\quad x_{ij}\in \{0,1\}, \forall i\in \mathcal{U}, ~ j\in \mathcal{B}\label{constraint:user_associate},\\
        &\quad d_{i}=\sum_{j\in \mathcal{B}}\sum_{k\in \mathcal{F}} \theta_{ik}x_{ij}(\tau_{k}+(1-y_{jk})w_{j}^{BH}),\\
        &\quad \sum_{j\in \mathcal{B}} x_{ij}=1 , \forall i\in \mathcal{U},\label{constraint:user association_constraint}
\end{align}
where $\alpha$ $\in [0,1]$ is reasonable in our paper and indicates the different significance between energy consumption and end-to-end file delivery delay. A larger $\alpha$ means that network operators will pay more attention to reducing energy consumption at the expense of increasing end-to-end file delivery delay. Constraint (\ref{constraint:power_constraint}) makes sure that total power supply does not exceed maximal power available at each SBS. Constraint (\ref{constraint:file datarate}) represents the transmission rate requirements of each file. Specifically, the extended expression of constraint (\ref{constraint:file datarate}) is
\begin{equation}\label{SINR}
  Wlog(1+\frac{p_{j}g_{ij}}{\sum\limits_{l\in \mathcal{B},l\neq j}p_{l}g_{il}+\sigma^{2}})\geq Wlog(1+x_{ij}\gamma_{k}\theta_{ik}),
\end{equation}
where $\gamma_{k}$ is SINR threshold that satisfies the rate requirement of file $f_k$. Constraint (\ref{constraint:file datarate}) can be rewritten as $\frac{p_{j}g_{ij}}{\sum\limits_{l\in \mathcal{B},l\neq j}p_{l}g_{il}+\sigma^{2}}\geq x_{ij}\gamma_{k}\theta_{ik}$. Thus, in stead of $R_k$, $\gamma_{k}$ can be used to represent the file transmission requirement. Each user association decision is indicated by a binary variable $x_{ij}$ and each user can only associates with one SBS, which are expressed as constraint (\ref{constraint:user_associate}) and (\ref{constraint:user association_constraint}).

The problem (\ref{tradeoff}) with discrete user association decision and continuous power control is an MIP problem. In order to solve the problem, file placement, user association and power control should be jointly considered, which also makes the problem much more complicated.

To further clarify the complexity of the problem $(\ref{tradeoff})$, we introduce a simple optimization instance. By the analysis of such instance, we show the challenge of solving (\ref{tradeoff}) and the proposition 1 is given below:
\begin{proposition}
When a file placement policy is chosen, for any feasible power allocation result, the problem $(\ref{tradeoff})$ is NP-hard.
\end{proposition}
\begin{proof}
In order to prove the proposition, we first introduce the well-known Multidimensional $0$-$1$ Knapsack Problems(MKP) with block angular structures which is a NP-Hard problem \cite{MKP}\cite{MKP1}.

In the problem MKP, there are $q$ knapsacks with a maximum weight load denoted by $(w_j,j=1,\cdots, q)$. And there are $n$ items. Each item has different values and weights in different knapsacks. Then the value and weight vectors of items in each knapsack can be denoted by $(\bm{v_j},j=1,\cdots, q)$ and $(\bm{b_j},j=1,\cdots, q)$, respectively. The MKP policy $(\bm{x_j}\in\{0,1\}^{n},j=1,\cdots, q)$ is to let each knapsack to select a subset of items, such that the total value of all knapsacks is maximized under limited weight load of each knapsack. The formulation os MKP is
\begin{align*}
  &\max_{\bm{x_1,\cdots,x_q}}~\bm{v_{1}^{T}x_1+\cdots+v_{q}^{T}x_q}\\
  &~~\text{s.t.}
    \begin{cases}
       & \bm{M_1x_1+\cdots+M_qx_q\preceq a_0}\\
       & \bm{b_{1}^{T}x_1} \leq w_1\\
       & \vdots\\
       & \bm{b_{q}^{T}x_q} \leq w_q\\
       & \bm{x_j}\in\{0,1\}^{n},j=1,\cdots, q,
    \end{cases}
\end{align*}
where $\bm{v_j}$ and $\bm{x_j}$ are $n$ dimensional value column vectors, $\bm{M_j}, j=1,\cdots, q$ are $m_0\times n$ coefficient matrices, and the first set of inequalities denotes $m_0$ coupling constraints. The constraint $\bm{b_{j}^{T}x_j} \leq w_j$ are $m_j$ dimensional block constraints, where $\bm{b_j}$ are $n$ dimensional weight column vectors. Then such problem can be viewed as a multidimensional $0$-$1$ knapsack problem with a block angular structure.

Back to problem $(\ref{tradeoff})$, a simple optimization instance is introduced. When the file placement and power allocation result are given, $(\ref{tradeoff})$ becomes a simple user association probelm. Let $(\hat{p}_{j}, j\in \mathcal{B})$ and $(\hat{y}_{jk},j\in \mathcal{B},k \in \mathcal{F})$ denote the file placement and power allocation results. Specifically, when user $u_i$ is connected to SBS $b_j$, the related energy consumption and end-to-end file delivery delay are $\hat{e}_{ij}=\hat{p}_{j}\sum_{k\in \mathcal{F}}\theta_{ik}\tau_{ij}^{k}$ and $ \hat{d}_{ij}=\sum_{k\in \mathcal{F}} \theta_{ik}(\tau_{ij}^{k}+(1-\hat{y}_{jk})w_{j}^{BH})$, respectively. According to $(\ref{shannon})$, for the given $(\hat{p}_{j}, j\in \mathcal{B})$ the data rate $r_{ij}$ and $\tau_{ij}^{k}$ are fixed. Thus, $\hat{e}_{ij}$ is a fixed value. Besides, as file placement policy $(\hat{y}_{jk},j\in \mathcal{B},k \in \mathcal{F})$ is known, according to $(\ref{delay})$, $\hat{d}_{ij}$ becomes a known end-to-end file delivery delay. Let cost coefficient $c_{ij}=\hat{e}_{ij}+\hat{d}_{ij}$ to denote the sum of $\hat{e}_{ij}$ and $\hat{d}_{ij}$. Base on the above analysis, we can rewrite our problem $(\ref{tradeoff})$ and obtain the below problem. To make the formulation more explicit, we intend to maximize the negative value of our problem.
\begin{align}\label{newobjective}
  &\max_{\bm{X}}~~~~-\sum_{i\in \mathcal{U}}\sum_{j\in \mathcal{B}}(c_{ij}x_{ij}) \\
  &~\text{s.t.}
  \begin{cases}
        &~0<\sum_{i\in \mathcal{U}}\hat{p}_jx_{ij}\leq P_{j}^{max},~ \forall j\in \mathcal{B}\nonumber\\
        &~\sum_{j\in \mathcal{B}} x_{ij}=1 , \forall i\in \mathcal{U},\nonumber\\
        &~x_{ij}\in \{0,1\}, \forall i\in \mathcal{U}, ~ j\in \mathcal{B},\nonumber
  \end{cases}
\end{align}

After comparing the problem $(\ref{newobjective})$ and the formulation of MKP, it is apparent that $(\ref{newobjective})$ is equivalent to the original instance of MKP. Therefore, the problem $(\ref{newobjective})$ is NP-Hard.

\end{proof}

\section{problem Solution}

The problem $(\ref{tradeoff})$ is difficult to be solved directly due to the coupling relationship among file placement, user association and power control. Then, we solve the problem in two stages, i.e. caching stage and delivery stage, based on the observation that caching is performed during off-peak time. At the caching stage, a local popular file placement policy is proposed by estimating user preference at each SBS. At the delivery stage, with given caching status at SBSs, the MIP problem is further decomposed by Benders' decomposition method.

\subsection{Local Popular File Placement Policy}
File processing consists of two stages, i.e., file caching and file delivery, which are implemented in different time scales. Different from the file delivery phase including the procedures of user association and power control, file caching is determined at a much slower time-scale. In file caching stage, based on the file popularity, files are often pre-fetched from the file server and proactively cached at SBSs during off-peak periods to alleviate traffic pressure on the backhaul link\cite{IVA-off-peak}. In some researches \cite{3u} \cite{IVA-timescale} \cite{IVA-timescale1}, authors studied the mixed-timescale problem: long-timescale file placement policy and the short-term user association and wireless resource allocation. However, in \cite{IVA-timescale1}, each BS has to cache the same files when they adopt the global file popularity-aware caching polity, which ignores the difference of the file popularity among small cells.  

In this paper, compared with the short-term user association and power control, the file placement policy is implemented during a longer periods $T$ such as some minutes or hours. During $T$, the popularity distribution of the files and the user preference for the files are fixed. For the convenience of this research and to avoid considering the tumultuous changes of the file popularity, the fixed file popularity of one $T$ is considered.

Based on the above analysis, we propose a local popular file placement policy. With the goal of maximizing the cache hit ratio, the local most popular files should be cached by each SBS during off-peak time. At a SBS, the local file popularity distribution can be obtained according to (\ref{filepopular}). Thus, the optimal file placement problem can be solved independently for each SBS. Considering SBS $b_j$, we can convert the optimal file placement problem to a knapsack problem, which is expressed as follows.

\begin{equation}\label{LPC}
  \begin{split}
    &\max_{\bm{y_j}}\quad  \psi_{j}= \sum_{k\in F}\psi_{jk}y_{jk} \\
    &~\text{s.t.}\quad~    \sum_{k\in F}s_{k}y_{jk}\leq C_j,~\forall j\in \mathcal{B} \\
    &\quad\quad\quad       y_{jk}\in\{0,1\},
  \end{split}
\end{equation}
where the binary variable $y_{jk}$ denotes the caching decision at $b_j$. As the knapsack problem is NP-Hard, a heuristic greedy algorithm for maximizing the caching hit probability is proposed and described as Algorithm 1.
\begin{algorithm}
  \caption{Greedy Algorithm for Maximum Caching hit Probability}
   \label{alg0}
   \KwIn{$~\mathcal{F},~B,~\psi_{jk},\forall k \in \mathcal{F}$.}
   \KwOut{$~\psi_{j}^{*},~\bm{y}_{j}^{*}, \forall j \in \mathcal{B}$.}

   \Repeat {$j > B$}{
   $j=1$\;
   Sort $\mathcal{F}$ into $\mathcal{F}_{j}^{\diamond}$ in descending order of $\frac{\psi_{jk}}{s_k}$\;
   Set $g \longleftarrow 0$, $k \longleftarrow 1$, and $y_{jk}^{\diamond}\longleftarrow 0$, $\forall k \in \mathcal{F}$\;
   \Repeat {$g > C_{j}$ and $k > F$}{
   Let $f_{k}^{'}$ be the $k$-th element of $\mathcal{F}_{j}^{\diamond}$\;
   Set $y_{jf_{k}^{'}}^{\diamond} \longleftarrow 1$\;
   Update $g \longleftarrow g+s_{k}$ and $k=k+1$\;
   }
   Calculate $\psi_{j}$ by (\ref{LPC}) using $\psi_{j}$\;
   $j=j+1$\;
   }
   \end{algorithm}

The algorithm requires $B$ iterations. The complexity of each iteration is $\mathcal{O}(C_j \log C_j)$.

\subsection{Association and Power Solution}
With the proposed popular file placement policy, the remaining problem (\ref{tradeoff}) is reduced to a MILP problem but still complicated with coupled user association and power control. To solve the MILP problem, Benders' decomposition is used to partition it into two small problems and obtain a $\epsilon$-optimal solution by iterations.

\subsubsection{Motivation}
To reduce the complexity of the problem (\ref{tradeoff}),  wireless transmission delay $\tau_{ij}^{k}$ in (\ref{delay function}) is relaxed to $\frac{s_k}{R_k}$. Then total wireless file transmission time $\sum_{i=1}^{U}\sum_{j=1}^{B}\sum_{k=1}^{F}\theta_{i_k}\tau_{k}x_{ij}$ becomes a constant $D$. This is because that for user $u_i$, $\sum_{k=1}^{F}\theta_{i_k}=1$ and $\sum_{j=1}^{B}x_{ij}=1$ hold. Thus, the system load among SBSs can be controlled by a load coefficient $\beta_{j}$ based on the capability of SBS $b_j(j\in \mathcal{B})$. And the load of SBS $b_j$ is can be expressed as $T_{j}=$D$\cdot\beta_{j}=\beta_{j}\sum_{i=1}^{B}\sum_{k=1}^{F}\theta_{i_k}\tau_{k}x_{ij}$.

With the above assumption and observation, given the proposed local popular file placement policy, the problem (\ref{tradeoff}) becomes an MILP problem. 
However, user association and power control is still coupled both in the objective function and the constraints, which make the problem complicated to be solved. Fortunately, based on the characteristics of our problem, the Benders' decomposition can be adopted to decomposition it.

Benders' decomposition is proposed for a class of MILP problems \cite{BD}\cite{GBD}. Instead of thinking about all variables of a problem, it first consider the continuous part. Thus, the original optimization problem is partitioned into two smaller problems: a subproblem with only continuous variables and a master problem with one continuous variable and multiple integer variables. To be specific, when integer variables are fixed, the resulting problem (subproblem) becomes a continuous linear program (LP) problem which can be solved by the standard duality theory of convex optimization. And then, the results of the dual problem can be transferred to the master problem.

\subsubsection{Subproblem}
According to the Benders' decomposition method, after the $(t-1)$th iteration, the energy consumption problem is formulated as a subproblem:
\begin{align}\label{sub}
    \min_{\bm{p}}  &\quad\sum_{j\in \mathcal{B}} e_{j} \\
  \text{s.t.}~ &\quad 0<\sum_{j\in \mathcal{B}}x_{ij}^{(t-1)}p_{j}\leq P_{j}^{max}, \forall j \in \mathcal{B},\\
       &\quad \frac{g_{ij}p_{j}+\varrho^{-1}(1-x_{ij}^{(t-1)})}{\sum\limits_{l\in \mathcal{B},l\neq j}p_{l}g_{il}+\sigma^{2}}\geq\gamma_{k}\theta_{ik}\label{constraint:equivalent datarate},~\forall i \in \mathcal{B},~ j \in \mathcal{U}.
\end{align}
Given the user association strategy, the power constraint is convex and hence will not change the nature of the formulated problem. To satisfy the standard problem form in Benders' decomposition, we use a equivalent transformation technique. A parameter $\varrho$ is introduced satisfying$\quad\varrho= \min\limits_{i}\frac{1}{\gamma_{f_{i}}((I-1)\bar{p}\bar{g}+\sigma^2)}$ where \(\bar{p}=\max \limits_{j}\{P_{j}^{max}\}\) and $\bar{g}=\max \limits_{i,j} \{g_{ij}\}$. The introduction of $\varrho$  will not change the optimal solution to the problem (\ref{tradeoff}). For $x_{ij}^{*}=1$, the formulation forms of (\ref{SINR}) and (\ref{constraint:file datarate}) are equivalent. For $x_{ij}^{*}=0$, from (\ref{constraint:file datarate}) we can deduce
\begin{equation}\label{proof}
  \begin{split}
& \frac{g_{ij}p_{j}+\varrho^{-1}}{\sum_{l\neq j}g_{il}p_l+\sigma^{2}}\;\geq \frac{((I-1)\bar{p}\bar{g}+\sigma^{2})\gamma_{fk}\theta_{ik}}{\sum_{l\neq j}g_{il}p_l+\sigma^{2}}\;\geq\gamma_{f_i},
  \end{split}
\end{equation}
No matter what $p_{ij}^{*}$ is, (\ref{proof}) always holds. In order to optimize energy consumption, the objective will make $p_{ij}^{*}$ be $0$.
And then based on the duality theory \cite{dual theory}, we can get the dual function of (\ref{sub}) as follows:
\begin{equation}\label{subproblem}
\begin{split}
\max_{\bm{\mu},\bm{\nu}} &\quad h(\bm{X}^{(t-1)},\bm{\mu},\bm{\nu}) \\
\text{s.t.} ~&\quad h(\bm{X}^{(t-1)},\bm{\mu},\bm{\nu})\\
       &= \sum_{j\in \mathcal{B}} (-P_{j}^{max}\mu_{j}) +\sum_{i\in \mathcal{U}}\sum_{j\in \mathcal{B}}(\varrho^{-1}(x_{ij}^{(t-1)}-1)+\sigma^{2}\gamma_{f_i})\nu_{ij},\\
       &e_{j}+\mu_{j}+\sum_{i \in \mathcal{U}}[ -g_{ij}\nu_{ij}+\sum_{l\in \mathcal{B}, l\neq j}(\gamma_{l}\theta_{il}g_{ij}\nu_{il})]\geq 0,\forall j\in \mathcal{B},\\
       &\bm{\mu}=[\mu_{j}]\succeq 0,~\forall j\in \mathcal{B},\\
       &\bm{\nu}=[\nu_{ij}]\succeq 0,~\forall j\in \mathcal{B},~ i\in \mathcal{U}.
\end{split}
\end{equation}
The dual function (\ref{subproblem}) is an LP problem, so Interior Point Method can be used to obtain the optimal solution \cite{Interior point}.
\subsubsection{Master Problem}
In the remaining problem, we mainly focus on end-to-end file delivery delay and low bound of energy consumption denoted by $\eta$, both of them depend on user association $\bm{X}$.
\begin{equation}\label{master_problem}
\begin{split}
  \min_{\eta,\bm{X}} & \quad\alpha\eta +(1-\alpha)\sum_{i \in \mathcal{U}}\sum_{j \in \mathcal{B}}\sum_{k \in \mathcal{F}}\theta_{ik}d_{ij}^{k}x_{ij} \\
  \text{s.t.} & \quad h(\bm{X},\bm{\mu}_{p}^{(m)},\bm{\nu}_{p}^{(m)})\leq \eta, \forall m=1,....,k_{1},\\
        & \quad h(\bm{X},\bm{\mu}_{q}^{(n)},\bm{\nu}_{q}^{(n)})\leq 0, \forall n=1,...,k_{2},\\
        & \quad x_{ij}\in \{0,1\}, \forall i\in \mathcal{U}, ~ j\in \mathcal{B},\\
        & \quad \sum_{j=1}^{B} x_{ij}=1 , \forall i\in \mathcal{U},
\end{split}
\end{equation}
where in the optimal cut $h(\bm{X},\bm{\mu}_{p}^{(m)},\bm{\nu}_{p}^{(m)})\leq \eta$, $(\bm{\mu}_{p}^{(m)},\bm{\nu}_{p}^{(m)})$ is the optimal solution of the bounded problem (\ref{subproblem}). And in the feasible cut $h(\bm{X},\bm{\mu}_{q}^{(n)},\bm{\nu}_{q}^{(n)})\leq 0$ $(\bm{\mu}_{q}^{(n)},\bm{\nu}_{q}^{(n)})$ is the unbounded direction of the unbounded problem (\ref{subproblem}). Both $(\bm{\mu}_{p}^{(m)},\bm{\nu}_{p}^{(m)})$ and $(\bm{\mu}_{q}^{(n)},\bm{\nu}_{q}^{(n)})$ form the constraint set of the problem (\ref{master_problem}). At $t$th iteration, $k_1$ and $k_2$ must satisfy: $k_1+k_2=t$.

\subsubsection{Upper and Lower Bounds}
The solutions of the problem (\ref{subproblem}) and (\ref{master_problem}) at each iteration  provide the upper and lower bounds of the optimal values respectively. Proposition \ref{ULB} as follows:
(UB and LB are denoted by $\Psi_{U}^{(t)}$ and $\Psi_{L}^{(t)}$)
\begin{proposition}\label{ULB}
At each iteration, the upper bounds $\Psi_{U}^{(t)}$ and lower bounds $\Psi_{L}^{(t)}$ are updated as follows: $\Psi_{L}^{(t)}$=$N^{(t)}$, and $\Psi_{U}^{(t)}$= $\min\limits_{0\leq r \leq t-1}\{M^{(r)}+\rho\sum_i\sum_j\sum_k \theta_{ik}d_{ij}^{k}x_{ij}^{(r)}\}$, where $M^{(t)}$ and $N^{(t)}$ are the optimal values of (\ref{subproblem}) and (\ref{master_problem}) at the $t$th iteration, respectively.
\end{proposition}

\begin{proof}

\noindent \textbf{\emph{Lower Bound:}}

First, we consider how to calculate lower bound $\Psi_{L}^{(t)}$ of the original problem (\ref{tradeoff}) at $t$-th iteration. The problem (\ref{subproblem}) is a dual function of the linear function (\ref{sub}). According to the strong duality of LP, we can say that the problem  (\ref{tradeoff}) is equivalent to that in (\ref{h(X)}):

\begin{align}\label{h(X)}
  \min_{\bm{X},\bm{\mu},\bm{\nu}} & \quad \alpha h(\bm{X},\bm{\mu},\bm{\nu}) + (1-\alpha)\sum_{i \in \mathcal{U}}\sum_{j \in \mathcal{B}}\sum_{k \in \mathcal{F}}  \theta_{ik}d_{ij}^{k}x_{ij} \\
   \text{s.t.} & \quad x_{ij}\in \{0,1\}, \forall i\in \mathcal{U}, ~j\in \mathcal{B},\\
        & \quad \sum\limits_{j=1}^{B} x_{ij}=1 , \forall i\in \mathcal{U}.
\end{align}

Compared (\ref{master_problem}) and (\ref{h(X)}), the relaxing constraints in (\ref{master_problem}) makes (\ref{master_problem}) is a relaxation of (\ref{h(X)}). At each iteration, a new constraint is added to the problem  (\ref{master_problem}). That is, the constraint set in the problem (\ref{master_problem}) will be updated after each iteration. According to the duality theory, this update of constraint set makes ($N^{(t)}=\alpha\eta^{(t)} +(1-\alpha )\sum_{i} \sum_{j} \theta_{ik}d_{ij}^{k}x_{ij}^{(t)}$) become lower bound of the optimal value in (\ref{h(X)}).

Then, $N^{(t)}$ is also the lower bound of the optimal values $\alpha\sum_j p_{j}^{*}\tau_{j} + (1-\alpha)\sum_{i} \sum_{j} \theta_{ik}d_{ij}^{k}x_{ij}^{*}$, where $(\bm{X}^{*},\bm{P}^{*})$ is assumed to be the optimal solution of the problem (\ref{tradeoff}).

Therefore, the optimal value $N^{(t)}$ of (\ref{master_problem}) at $t$-th iteration is a lower bound $\Psi_{U}^{(t)}$ of problem (\ref{tradeoff}).

\noindent \textbf{\emph{Upper Bound:}}

We prove that $\min\limits_{0\leq r \leq t-1}\{\alpha M^{(r)}+(1-\alpha)\sum_i\sum_j\theta_{ik}d_{ij}^{k}x_{(ij)}^{(r)}\}$ is the upper bound of the problem (\ref{tradeoff}) at the $t$-th iteration of .

As $\bm{y}^{(t-1)}$ makes the problem (\ref{subproblem}) either bound or unbound, the optimal value $M^{(r)}$ of (\ref{subproblem}) will be either finite or infinite, respectively. If $M^{(r)}$ is infinite, it is apparent that $\min\limits_{0\leq r \leq t-1}\{\alpha M^{(r)}+(1-\alpha)\sum_i\sum_j\sum_k \theta_{ik}d_{ij}^{k}x_{(ij)}^{(r)}\}$ is the upper bound. If $M^{(r)}$ is finite, $\omega=\arg\min\limits_{0\leq r \leq t-1}\{\alpha M^{(r)}+(1-\alpha)\sum_i\sum_j\theta_{ik}d_{ij}^{k}x_{(ij)}^{(r)}\}$, where $0\leq\omega\leq t-1$. Correspondingly, $(\bm{X}^{(\omega)},\bm{\mu}^{(\omega)},\bm{\nu}^{(\omega)})$ and power $\bm{p}^{(\omega)}$ are the optimal solution of $h(\bm{X},\bm{\mu},\bm{\nu})$ and (\ref{sub}). According to the strong duality, we have $M^{(\omega)}=h(\bm{X}^{(\omega)},\bm{\mu}^{(\omega)},\bm{\nu}^{(\omega)})
=
\sum\limits_{i=1}^{U}\sum\limits_{j=1}^{B}p_{j}^{(\omega)}\tau_{j}$.
If we assume $\alpha\sum\limits_{j=1}^{B}p^{(\omega)}\tau_{j}+(1-\alpha)
\sum\limits_{i=1}^{U}\sum\limits_{j=1}^{B}\theta_{ik}d_{ij}^{k}x_{ij}^{(\omega)}$ is less than $\alpha\sum\limits_{j=1}^{B}p^{(*)}+(1-\alpha)\sum\limits_{i=1}^{U}\sum\limits_{j=1}^{B}\theta_{ik}d_{ij}^{k}x_{ij}^{*}$, then $(\bm{P}^{(\omega)},\bm{X}^{(\omega)})$ will be the optimal solution of the problem (\ref{tradeoff}). It means after $\omega$th iteration $\Psi_{U}^{(t)}< \Psi_{L}^{(t)}$,  which is contradictory.
Hence, $\min\limits_{0\leq r \leq t}\{\alpha M^{(r)}+(1-\alpha)\sum_i\sum_j\sum_k \theta_{ik}d_{ij}^{k}x_{ij}^{(r)}\}$ is the upper bound $\Psi_{U}^{(t)}$ of the problem (\ref{tradeoff}).
\end{proof}

\subsubsection{Relaxed Master Problem (RMP)}
Considering the binary nature of $x_{ij}$, which domains the computation complexity of the problem (\ref{tradeoff}), we decide to use a linear relaxation method. Firstly, instead of $x_{ij}\in\{0,1\}$, we make $x_{ij}\in[0,1]$ by Proposition 3. Then we construct an equivalent formulation with a penalty function by Proposition 4, which can reduce the computation complexity of the problem (\ref{tradeoff}).

In Proposition 3, the equivalence relationship between the binary constraint and the linear relaxation is elaborated.

\begin{proposition}\label{pro set}
Given the definitions
\begin{subequations}
\begin{equation}A:=[0,1]^{UB},\end{equation}
\begin{equation}B:=\Bigg\{\bm{x}\in R^{UB}:\sum_{i \in \mathcal{U}}\sum_{j \in \mathcal{B}}x_{ij}^{2}-\sum_{i \in \mathcal{U}}\sum_{j \in \mathcal{B}}x_{ij}<0\Bigg\},\end{equation}
\end{subequations}
the binary set $\{0,1\}^{UB}$ is the difference of two convex sets A
and B, i.e., $\{0,1\}^{UB} = A\setminus B $.
\end{proposition}
\begin{proof}
Obviously, we can get $\{0,1\}^{UB} \subset A\setminus B $. Besides, $x_{ij}\in \{0,1\}^{UB}$ is the result of
\begin{equation}\label{set}
  x_{ij}-x_{ij}^{2}=0, ~i=1, ..., U;j=1, ..., B.
\end{equation}

Then, $x_{ij}-x_{ij}^{2}\geq0$ holds for each $x_{ij}\in A$ and $\sum_{i=1}^{U}\sum_{j=1}^{B}x_{ij}-\sum_{i=1}^{U}\sum_{j=1}^{B}x_{ij}^{2}\leq0$ for $x_{ij} \notin B$, so each $x_{ij}\in A\setminus B $ makes $x_{ij}-x_{ij}^{2} = 0$.  Therefore $x_{ij}\in A\setminus B $ is feasible to (\ref{set}), i.e., $A\setminus B\subset\{0,1\}^{UB}$.
\end{proof}

Based on Proposition 3, (\ref{master_problem}) can be equivalently transformed to (\ref{master_problem2}). The detailed proof is given in Proposition 4.
\begin{proposition}\label{pro_master_problem2}
We can relax the binary set $\bm{X}$ in (\ref{master_problem}) to $[0,1]^{UB}$ and obtain an new objective function (\ref{master_problem2}) which is equivalent to (\ref{master_problem}) when $\lambda\gg1$.
\begin{equation}\label{master_problem2}
\begin{split}
\textbf{RMP:}\\
\min_{\eta,\bm{X}}&\quad\alpha\eta+(1-\alpha)\sum_{i \in \mathcal{U}}\sum_{j \in \mathcal{B}}\sum_{k \in \mathcal{F}} \theta_{ik}d_{ij}^{k}x_{ij}\\
&\quad+\lambda\sum_{i=1}^{U}\sum_{j=1}^{B}(x_{ij}-x_{ij}^{2})\\
   s.t. & \quad h(\bm{X},\bm{\mu}_{p}^{(m)},\bm{\nu}_{p}^{(m)})\leq \eta, \forall m=1, ..., t_{1},\\
        & \quad h(\bm{X},\bm{\mu}_{q}^{(n)},\bm{\nu}_{q}^{(n)})\leq 0, \forall n=1, ..., t_{2},\\
        & \quad x_{ij}\in [0,1], \forall i\in U, ~j\in B,\\
        & \quad \sum_{ij} x_{ij}=1 , \forall j\in B,
\end{split}
\end{equation}
where $\lambda$ is a constant penalty factor. The large parameter $\lambda$ makes the relaxed $\bm{X}$ be as binary as possible.
\end{proposition}

\begin{proof}
See Appendix.
\end{proof}

Based on Proposition \ref{pro_master_problem2}, when an appropriate value is chosen for $\lambda$, the problem (\ref{master_problem}) is equivalent to the problem (\ref{master_problem2}) in the sense that they share the same optimal values as well as optimal solution. The RMP is a minimization of a concave quadratic function which can be solved by the method in\cite{concavequadratic}.

\subsection{Algorithm }
In order to get an $\epsilon$-optimal value, we propose a  user association and power control (UCWT) algorithm in Algorithm \ref{alg1}.
%
%
%
%

\begin{algorithm}
  \caption{User assoCiation and poWer conTrol (UCWT) Algorithm for Energy-Delay Tradeoff}
   \label{alg1}
   \KwIn{$~\bm{P}^{max},~\bm{\gamma},~\bm{\theta}$}
   \KwOut{$~\bm{P}^*,~\bm{X}^*$}
   \textbf{Initialization:}~Let$~\bm{X}^{(0)}=\bm{0},~\Psi_{L}^{(0)}=-\infty, ~\Psi_{U}^{(0)}=+\infty,~t=1,~m=1,~n=1$\;

   \Repeat{$~\Psi_{U}^{(t-1)}-~\Psi_{L}^{(t-1)}\leq \epsilon$}{
   \textbf{Subproblem:}\\
   Solve $(\bm{\mu}^{(t)},\bm{\nu}^{(t)})=\arg\max\limits_{(\bm{\mu},\bm{\nu})} h(\bm{X}^{(t-1)},\bm{\mu},\bm{\nu})$ in (\ref{subproblem}) with Interior Point Method\;
   \eIf{\rm{ (\ref{subproblem}) is bounded}}{
   Get extreme point:
   $(\bm{\mu}_{p}^{(m)},\bm{\nu}_{p}^{(m)})=(\bm{\mu}^{(t)},\bm{\nu}^{(t)})$;
   }
   {
   Get extreme ray:
   ($\bm{\mu}_{q}^{(n)},\bm{\nu}_{q}^{(n)})=(\bm{\mu}^{(t)},\bm{\nu}^{(t)})$;}
   Calculate upper bound $\Psi_{U}^{(t)}$ with
   $\omega=\arg\min\limits_{0\leq r \leq t-1}\{\alpha M^{(r)}+(1-\alpha)\sum\limits_i\sum\limits_j\theta_{ik}d_{ij}^{k}x_{ij}^{(r)}\}$\;

   {
   \textbf{RMP:}\\
    Add a constraint: $h(\bm{X},\bm{\mu}_{p}^{(m)},\bm{\nu}_{p}^{(m)})\leq\eta$ or $h(\bm{X},\bm{\mu}_{q}^{(n)},\bm{\nu}_{q}^{(n)})\leq0$ to (\ref{master_problem2})\;
   Solve (\ref{master_problem2}) to obtain $\bm{X}^{(t)}$\;
   Calculate lower bound with the method in \cite{concavequadratic}: $\Psi_{L}^{(t)}=\alpha\eta^{(t)}+(1-\alpha)\sum_{i=1}^{U}\sum_{j=1}^{B}\sum_{k=1}^{F} \theta_{ik}d_{ij}^{k}x_{ij}^{(t)}$\;

   $t=t+1,n=n+1,m=m+1$\;}
   }
   Get optimal solution $\bm{P}^{*}$ through solving the dual problem of (\ref{subproblem}) with $\bm{X}^{(\omega)}$. Let optimal solution $\bm{X}^{*}=\bm{X}^{(\omega)}.$
   \end{algorithm}

\subsubsection{Analysis of the Convergency}
The UCWT algorithm uses the gap between $\Psi_{U}^{(t)}$ and $\Psi_{L}^{(t)}$ as the termination criterion. In particular, if the gap is equal to zero the exact global optimal solution to the problem (\ref{tradeoff}) . If the gap is equal to $\epsilon$, an $\epsilon$-optimal value is obtained. The following
theorem proves the convergence of the proposed algorithm.

\newtheorem*{Theorem 1}{Theorem 1}
\begin{Theorem 1}
After a finite number of iterations, the proposed UCWT algorithm converges to a global optimal value of (\ref{tradeoff}).
\end{Theorem 1}
\label{convergency}
\begin{proof}
After each iteration, the constraint set of the problem (\ref{master_problem2}) is updated by adding $ h(\bm{X},\bm{\mu}_{p}^{(m)},\bm{\nu}_{p}^{(m)})\leq\eta$ or $h(\bm{X},\bm{\mu}_{q}^{(n)},\bm{\nu}_{q}^{(n)})\leq0$, where $(\bm{\mu}^{(m)},\bm{\nu}^{(n)})$ or $(\bm{\mu}^{(n)},\bm{\nu}^{(n)})$ is one feasible solution of the subproblem (\ref{subproblem}). According the linear programming theory, the solution set of the problem (\ref{subproblem}) is a finite number of the extreme points or extreme rays which determine the total iteration number. After finite iterations, the constraint set of the subproblem (\ref{subproblem}) is completed by adding the total extreme points or extreme rays of (\ref{subproblem}). This implies that  optimal user association ($\bm{X}^{*}$) can be obtained by (\ref{master_problem}). As $\Psi_{L}^{(t)}$ is increasing after each iteration $t$, the $\Psi_{L}$ will arrive at the optimal value of problem (\ref{tradeoff}). For $\Psi_{U}^{(t)}$ is also decreasing after each iteration $t$, $\Psi_{U}^{(t)}$ finally satisfies the optimal value of (\ref{tradeoff}) with the optimal solution ($\bm{X}^{*}$) and ($\bm{P}^{*}$) solved by (\ref{subproblem}). At last, the gap of $\Psi_{L}^{(t)}$ and $\Psi_{U}^{(t)}$ is shrunk to 0 within a finite number of iterations.
\end{proof}

\subsubsection{Analysis of the Complexity}
For the subproblem (\ref{subproblem}), it is an LP problem and Interior Point Method has been proposed as an efficient method in \cite{Interior point}. Here, Interior Point Method needs be executed $\mathcal{O}(U^{2}B^{2})$ times and at each time the complexity is $\mathcal{O}(U^{2}B)$. So the toatl complexity of subproblem is $\mathcal{O}(U^{4}B^{3})$.

For the RMP (\ref{master_problem2}), it is a concave quadratic function which can be solved by the method in \cite{concavequadratic}. It needs $t^2$ iterations and each iteration complexity is $\mathcal{O}(U^{2}B)$. So the toatl complexity of RMP is $\mathcal{O}(t^2U^{2}B)$.

According to Theorem 1, the total number of iterations in UCWT, relate to the finite number of extreme points or extreme rays in subproblem (\ref{subproblem}), is determined by the scale of the problem, i.e., total user number $U$ and total SBS number $B$. Thus, the total iterations in UCWT will be varied under different value of $U$ and $B$. In our simulations, UCWT usually obtains a $\epsilon-$optimal value after tens of iterations.

\section{Performance Evaluation}
Extensive simulations are carried out to validate our work. The results demonstrate the convergency of the proposed algorithm. Moreover, with the proposed algorithm, the desired energy-delay tradeoff can be obtained under various scenarios in cache-enabled DSCNs.

\begin{table}
  \centering
  \caption{SIMULATION SETTING}
  \label{simulationsetting}
  \begin{tabular}{|c|c|}
    \hline
    \textbf{Symbol}                              & \textbf{Description}\\\hline
    Small cell radius                            & 40m \\\hline
    Maximal transmit power of each SBS            & 23dBm \\\hline
    Number of subchannels                        &16\\\hline
    Subchannel bandwidth                         &200KHz\\\hline
    Thermal noise density&-174dBm/Hz\\\hline
    Number of files& 600\\\hline
    Size of file& 0.5MB to 50MB\\\hline
  \end{tabular}
\end{table}

\subsection{Simulation Settings}
In the simulations, we study a cache-enabled DSCN consisting of $25$ small cells in a 250m-by-250m square area. We assume that SBSs are uniformly distributed over a two-dimensional network layout and each SBS is located at the center of its serving cell. Users are randomly generated and deployed according to the uniform distribution. The radius of the central area is $25$ meters\cite{radius}. We consider a distance dependent path loss model and the loss factor from SBS $b_j$ to user $u_i$ is given as $d_{ij}^{(-\kappa)}$($2\leq\kappa\leq 5$). The physical layer parameters are based on the 3GPP evaluation methodology document\cite{3GPP}. We use a file library of 600 files, the size of which follows the  uniform distribution between [0.5, 50](MB). The SINR requirement of each file is set between 1.5 and 5, which can be converted to rate requirement according to (\ref{SINR}). We assume that the probability that each user generates a request is equal, namely $p(u_{i})=\frac{1}{|\Phi_{j}|}$. Therefore, Eq. (\ref{filepopular}) can be rewritten as: $\psi_{jk}=\sum_{i\in \Phi_{j}} \frac{1}{|\Phi_{j}|}\rho_{ik}$, where $|\Phi_{j}|$ is the cardinality of the set $\Phi_{j}$. The penalty parameter $\lambda$ for the proposed algorithm is set to $(10P^{max}+\sum_{i=1}^{B}\sum_{k=1}^{F}\theta_{i_k}\tau_{k}+\sum_{j=1}^{B}w_{j}^{BH})$, such that the value of the penalty
term $\lambda\sum_{i=1}^{U}\sum_{j=1}^{B}(x_{ij}-x_{ij}^{2})$is comparable to the value of (\ref{master_problem}) \cite{penalty}. Given that each SBS has the same maximal transmission power, load coefficient $\beta_j$ defined in Section \uppercase\expandafter{\romannumeral3}-A is set to be $1/B$. Some major parameters are configured in Table \ref{simulationsetting}.

In Algorithm \ref{alg0}, each user preference follows normal distribution with different mean value (1$\sim$600) and variances over the total files, which is implemented at the beginning of the simulation. For each SBS, the radius of the center area is set as $25m$. The central user requests for file $f_k$ at SBS $b_j$ must satisfy :$ \frac{\sum_{i \in \Phi_{j}}{\theta_{ik}}}{|\Phi_{j}|}=\psi_{jk}$.

\subsection{Convergence Analysis of UCWT}
\begin{figure}[htbp]
\centering
\includegraphics[width=0.3\textwidth]{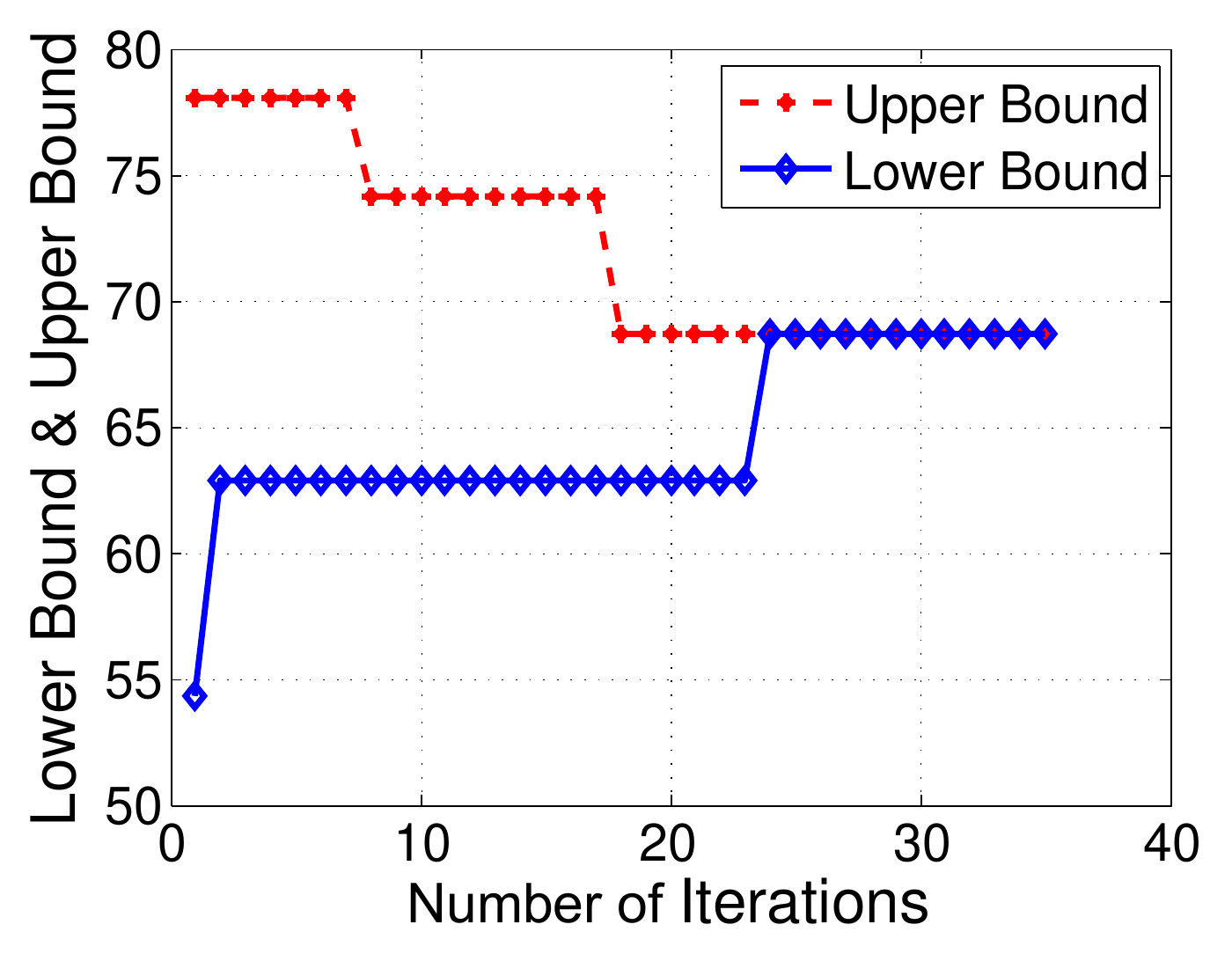}
\caption{Upper and the lower bounds when performing UCWT}
\label{iteration}
\end{figure}

In order to verify the convergency of the proposed UCWT algorithm, the number of iterations for an optimal value is plotted in Fig. \ref{iteration}.  The cache capacity of each SBS follows a normal distribution with  mean value(15 files) and the number of users is 150. Total user requests for files in central area of $b_j$ follows the distribution $\psi_{jk}, \forall k \in \mathcal{F}$. From Fig. \ref{iteration}, the upper bound and lower bound become closer with increasing the number of iterations. Through limited number of iterations, the gap between UB and LB converges to a given $\epsilon$. We observe that, if the number of users is not very large, UCWT can converge to an $\epsilon$-optimal value after tens of iterations, which is due to the fact that the number of the extreme points or rays is limited in this case. 

\subsection{Energy-delay Tradeoff by Adjusting $\alpha$}

%
%
%
%

\begin{figure}[ht]

\centering                                            
\subfigure[]{\includegraphics[width=4.4cm]{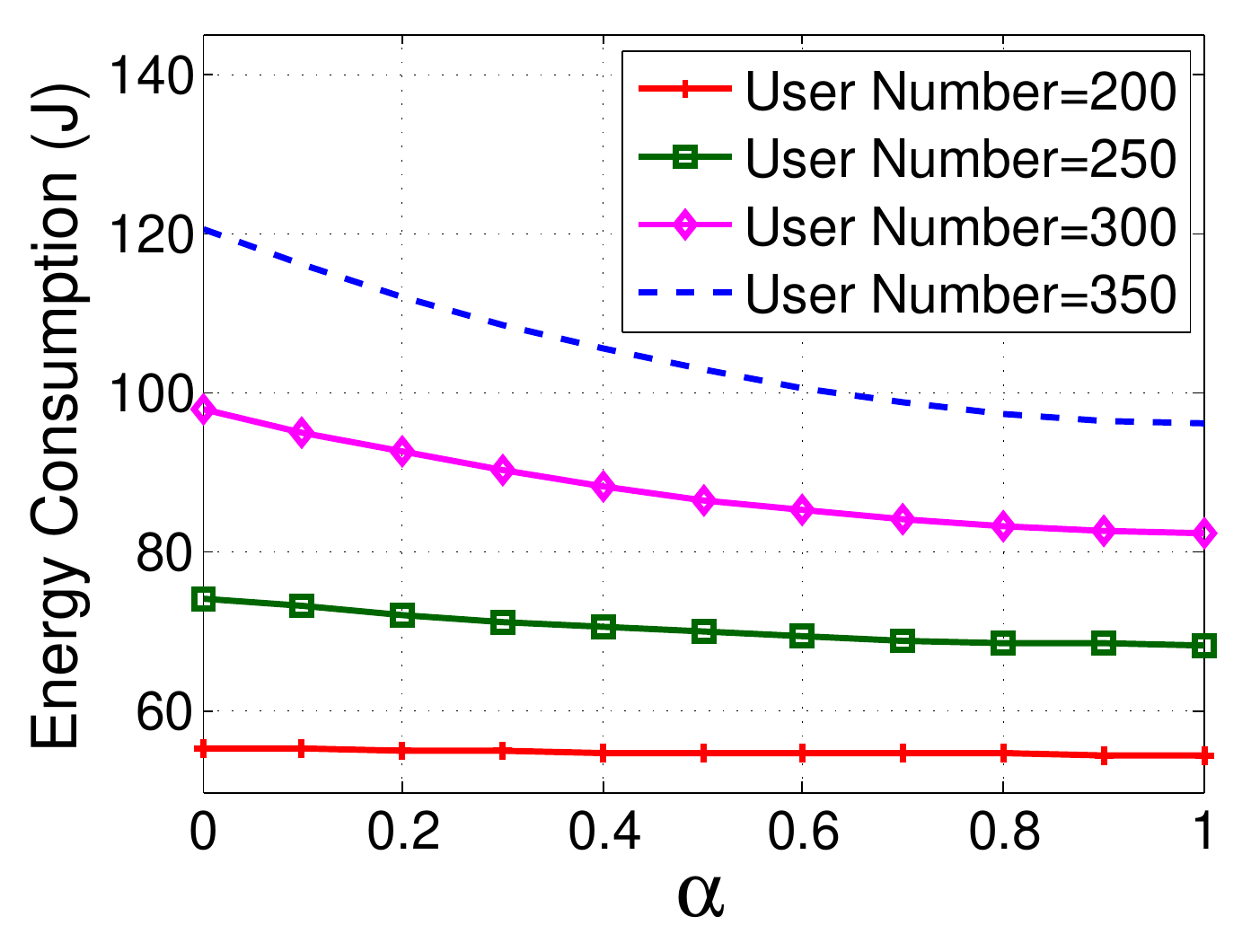}\label{UserEnergyTradeoff}}
\hspace{-2.3ex}
\subfigure[]{\includegraphics[width=4.4cm]{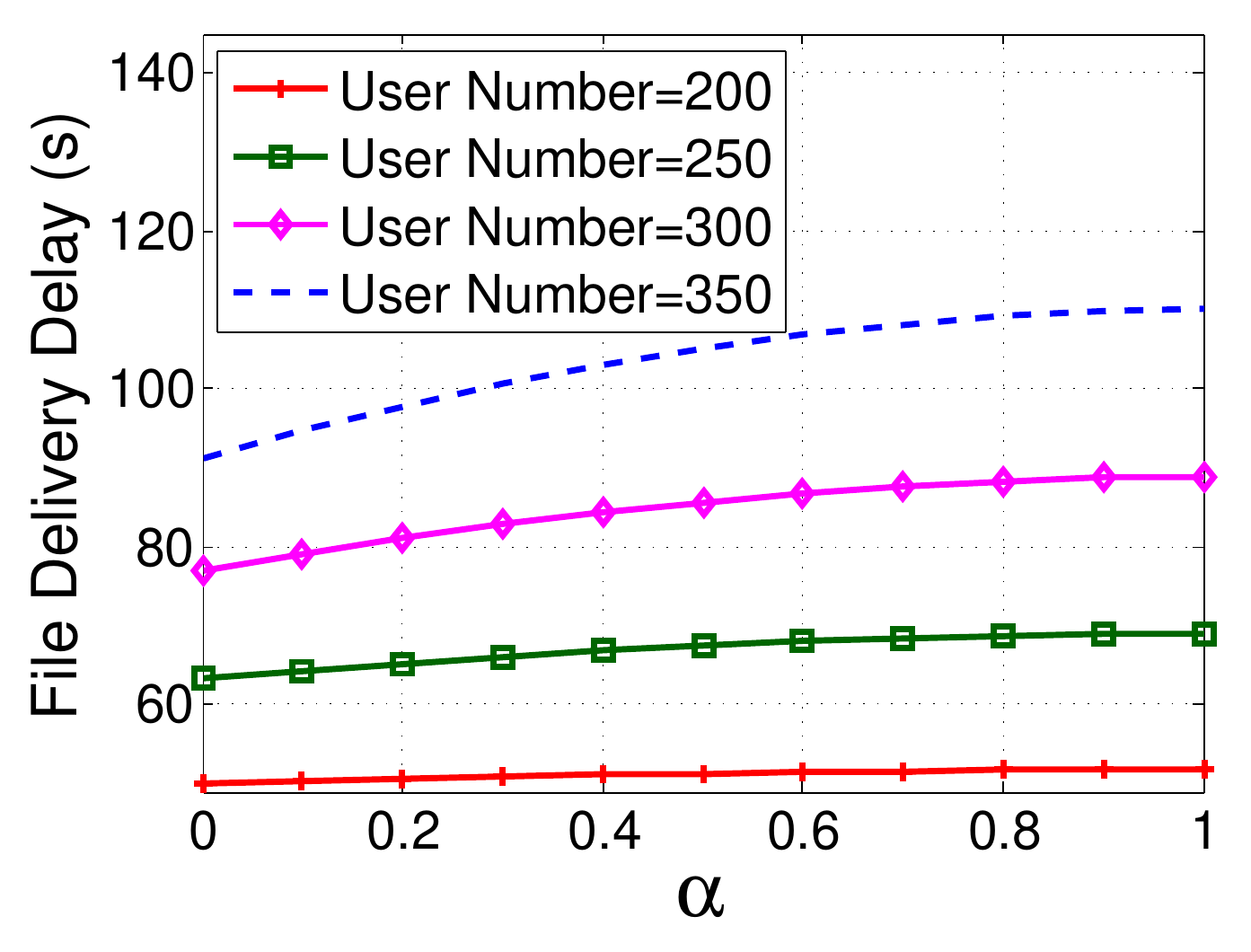}\label{UserDelayTradeoff}}                                        
\caption{(a) energy consumption and (b) end-to-end file delivery delay comparison under different number of users by varying $\alpha$}

\label{UserNumberTradeoff}
\end{figure}

\begin{figure}[ht]

\centering                                            
\subfigure[]{\includegraphics[width=4.3cm]{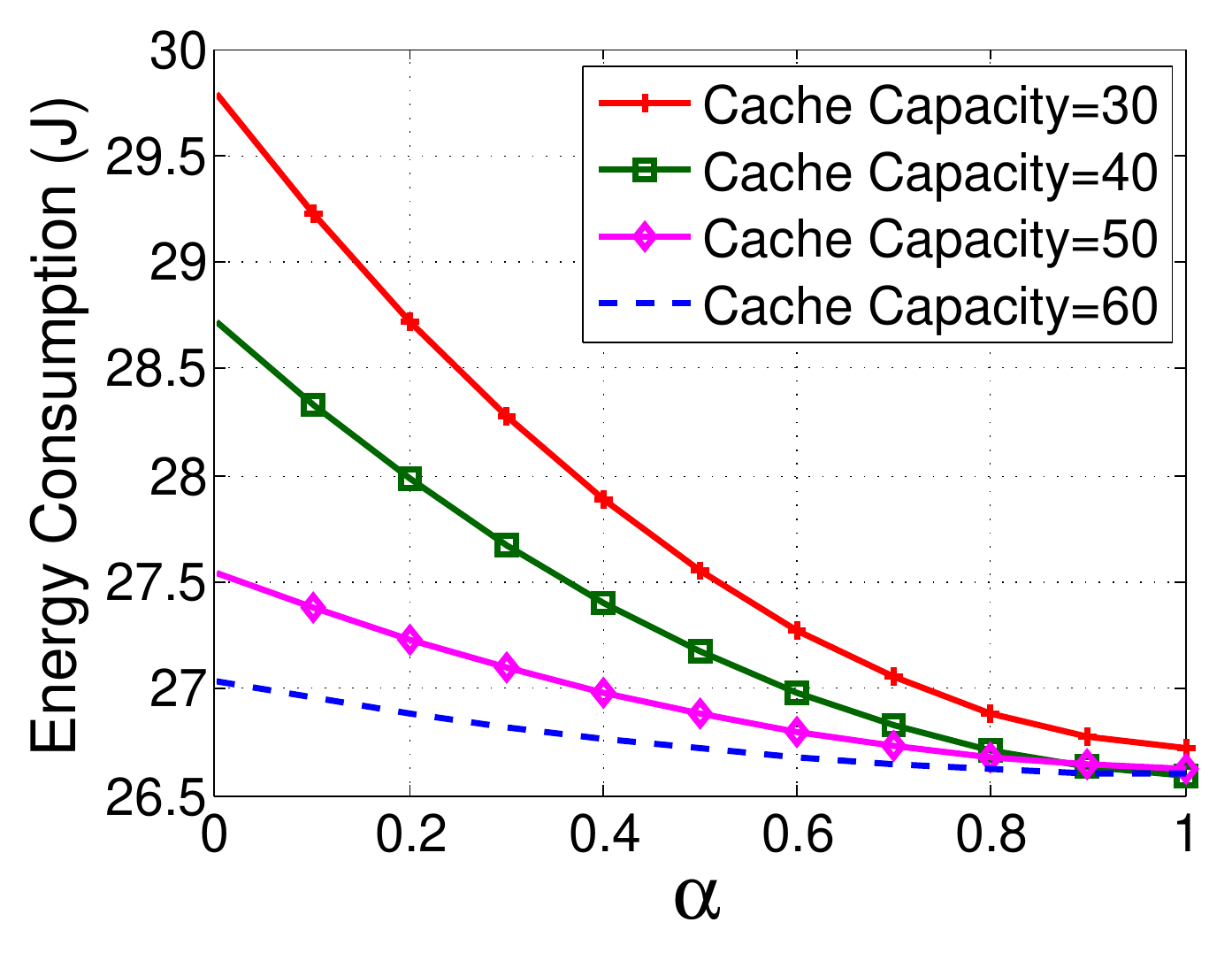}\label{CapacityEnergyTradeoff}}
\hspace{-1.3ex}
\subfigure[]{\includegraphics[width=4.3cm]{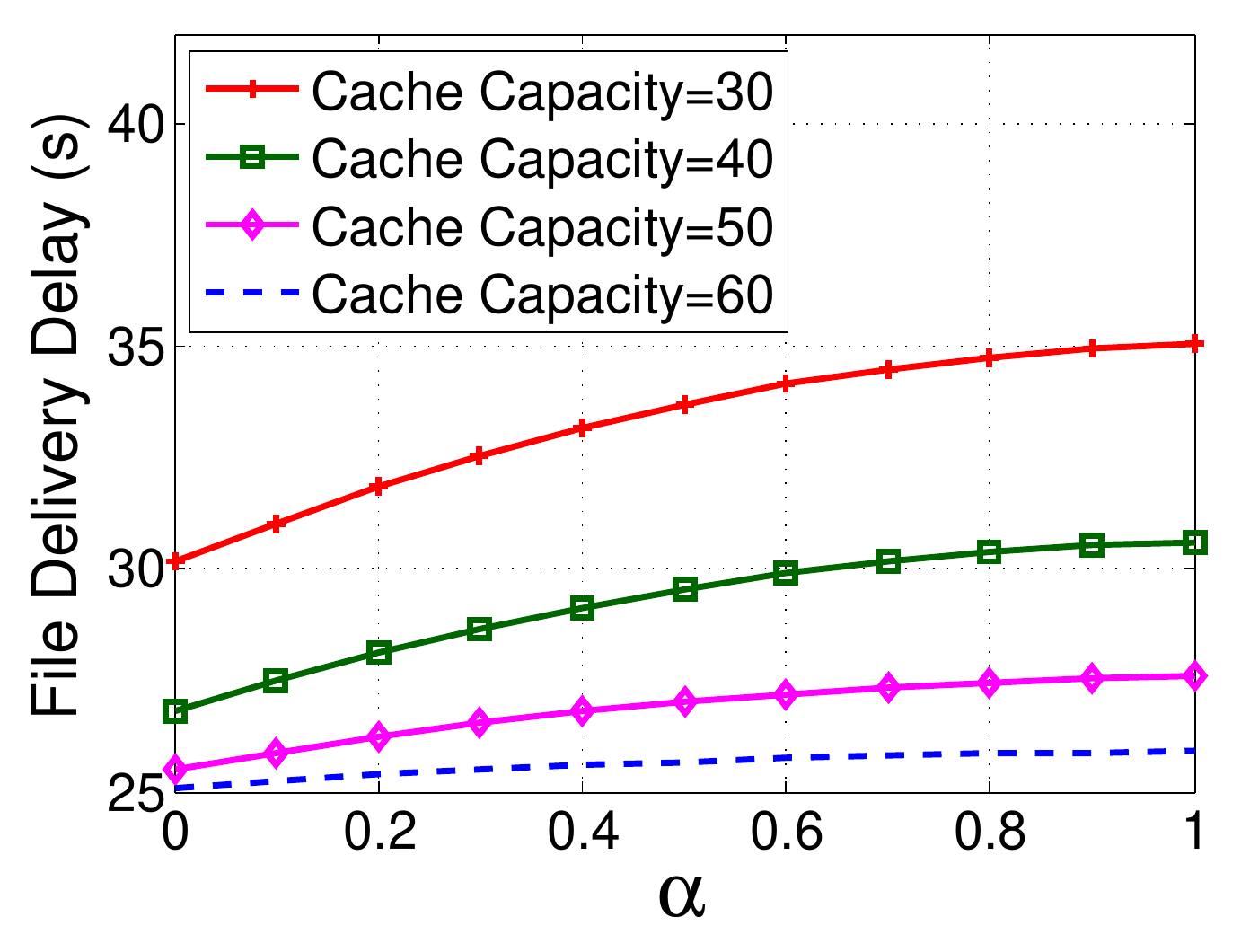}\label{CapacityDelayTradeoff}}                                        
\caption{(a) energy consumption and (b) end-to-end file delivery delay comparison under different average cache capacities by varying $\alpha$}

\label{CapacityDelayTradeoff}
\end{figure}

We investigate the tradeoff characteristics between energy consumption and end-to-end file delivery delay by varying $\alpha$. Fig. \ref{UserNumberTradeoff} shows the energy-delay tradeoff curves by adjusting  $\alpha$ from $0$ to $1$ given average capacity 15(number of files). We can see that, when the number of users is given, as the $\alpha$ is increasing, energy consumption is in a decreasing trend (Fig. \ref{UserEnergyTradeoff}) and end-to-end file delivery delay in a increasing trend (Fig. \ref{UserDelayTradeoff}). This is because when $\alpha$ increases, more users are forced to associate with the nearer SBSs regardless of their file requests. Thus lower energy consumption is achieved. Besides, by increasing energy consumption by $20\%$, end-to-end file delivery delay can be reduced by an average $10\%$ with an proper $\alpha$.

Fig. \ref{CapacityDelayTradeoff} shows the same tradeoff characteristics between energy consumption and end-to-end file delivery delay under different cache capacities given the number of user 150. Note that, when the average cache capacity is large enough, energy consumption and end-to-end file delivery delay are very small and change a little with various $\alpha$. This is due to the fact that most of the required files can be cached in the nearest SBSs, which saves a lot energy and delay.

\subsection{Effects of Caching Strategies}


\begin{figure}[ht]

\centering                                         
\subfigure[]{\includegraphics[width=4.4cm]{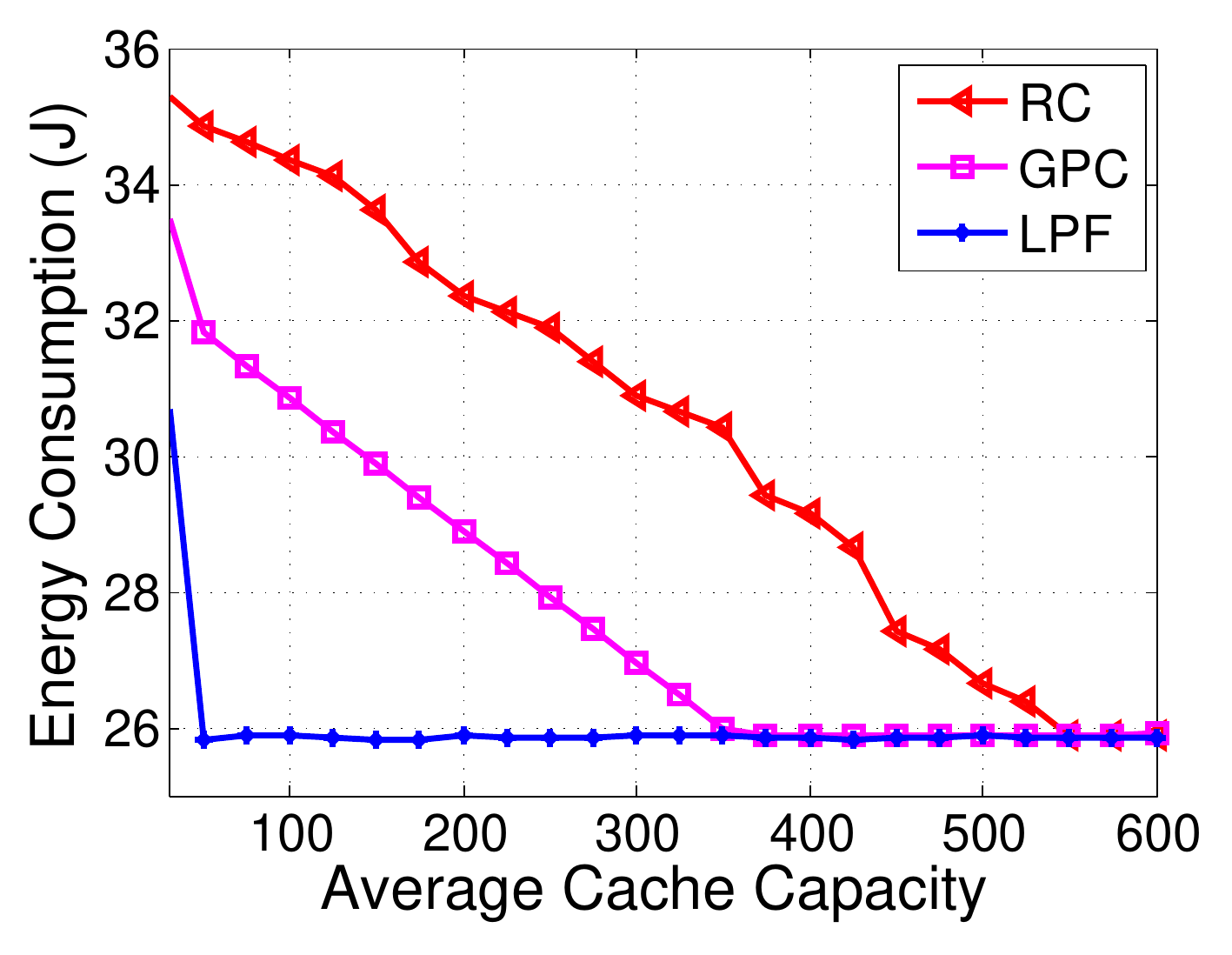}\label{EnergyCacheStrategy}}
\hspace{-2.3ex}
\subfigure[]{\includegraphics[width=4.4cm]{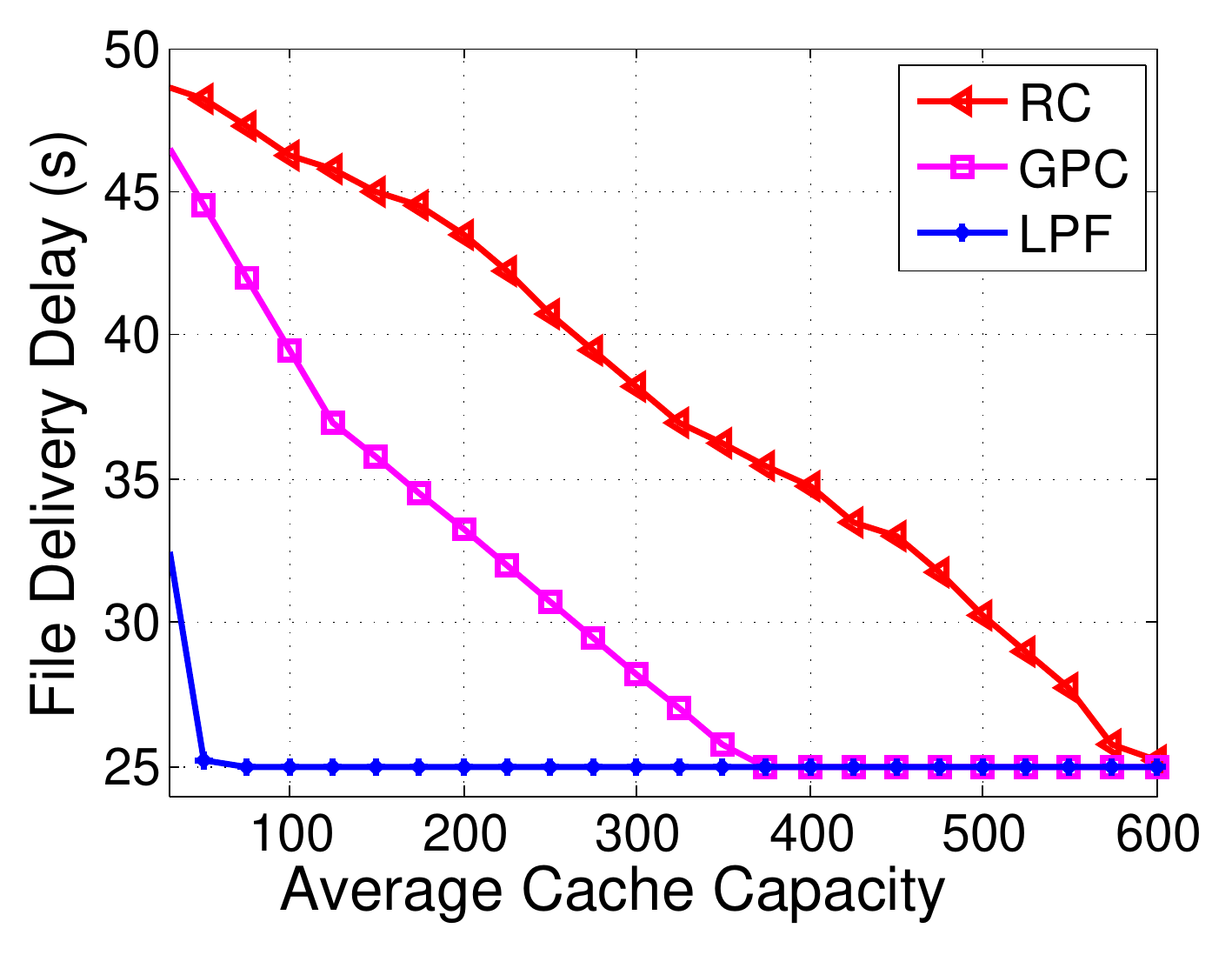}\label{DelayCacheStrategy}}                                        
\caption{(a) energy consumption and (b) end-to-end file delivery delay comparison of caching strategies under different average cache capacities}
\label{Caching Strategies}
\end{figure}

We compare the proposed local popular file placement policy (LPF) with the following caching policies:
\begin{itemize}
  \item \emph{Global Popularity Caching (GPC):} It is assumed that the file popularity in DSCNs follows a global popularity distribution. That's to say, all SBSs will cache the global popular files and content in their cache is the same without consideration of cache capacity.
  \item \emph{Random Caching (RC):} Each SBS randomly chooses the files to cache regardless of the file popularity distribution.
  \end{itemize}

In Fig. \ref{Caching Strategies}, we compare the performance of different caching policies under different average cache capacities at SBSs. The number of user is 150 and $\alpha$ is 0.5. As the average cache capacity increases, LPF achieves the best performance among all three policies. This is due to the fact that the file popularity distribution in different small cells may be different from each other, and the global file popularity distribution does not reflect local file popularity. In detail, for example, in Fig. \ref{EnergyCacheStrategy} when the average cache capacity measured in the number of files is beyond $50$, the energy consumption value obtained by LPF will not decrease. The reason is that each SBS has enough storage to cache all files requested by its local users. But for GPC, each SBS needs much more storage to cache all global popular files (about $350$) to satisfy all requests from the local users.

\subsection{Performance Comparison}

\begin{figure}[ht]

\centering                                            
\subfigure[]{\includegraphics[width=4.4cm]{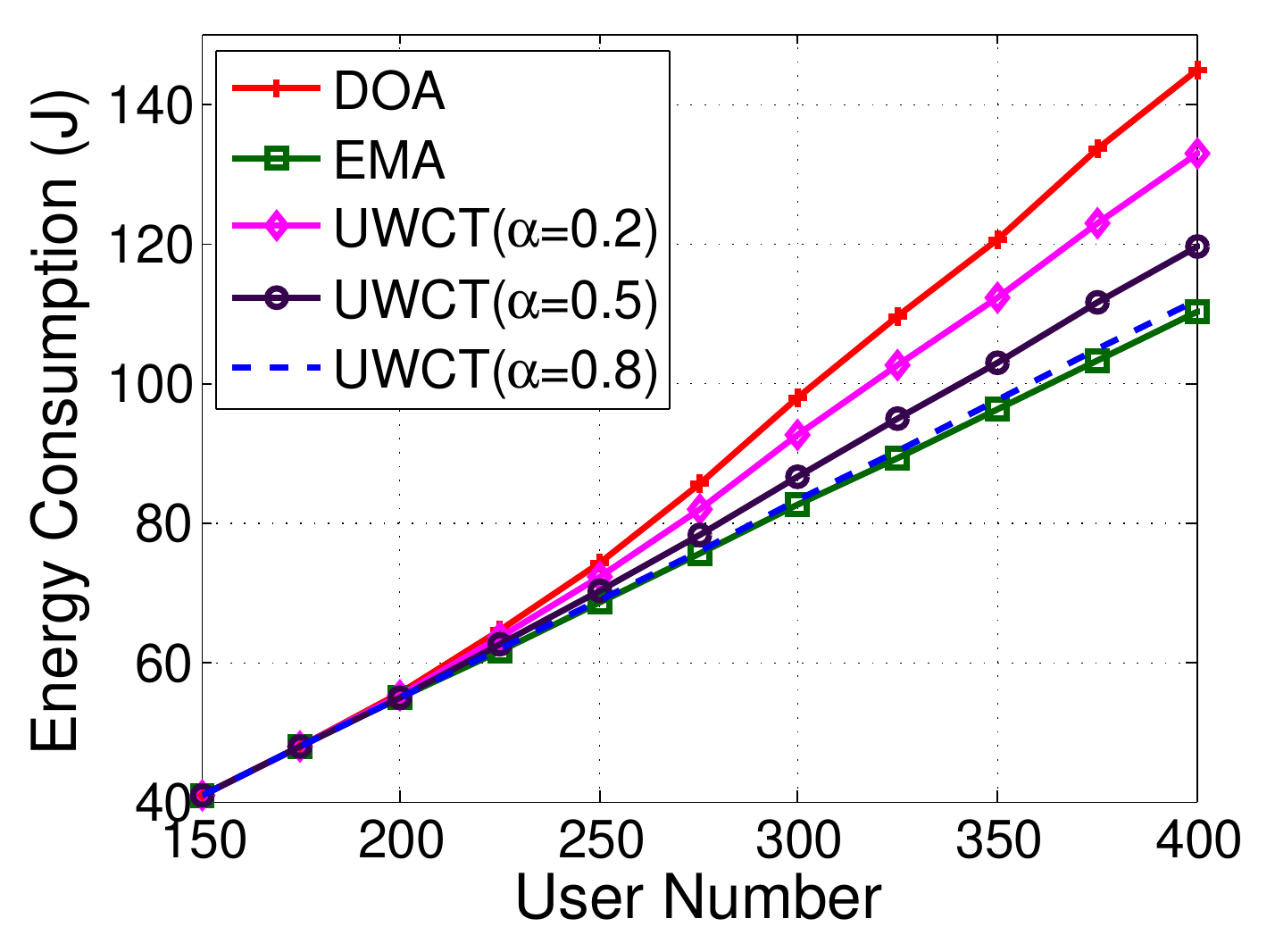}\label{UserEnergyComparision}}
\hspace{-2.3ex}
\subfigure[]{\includegraphics[width=4.4cm]{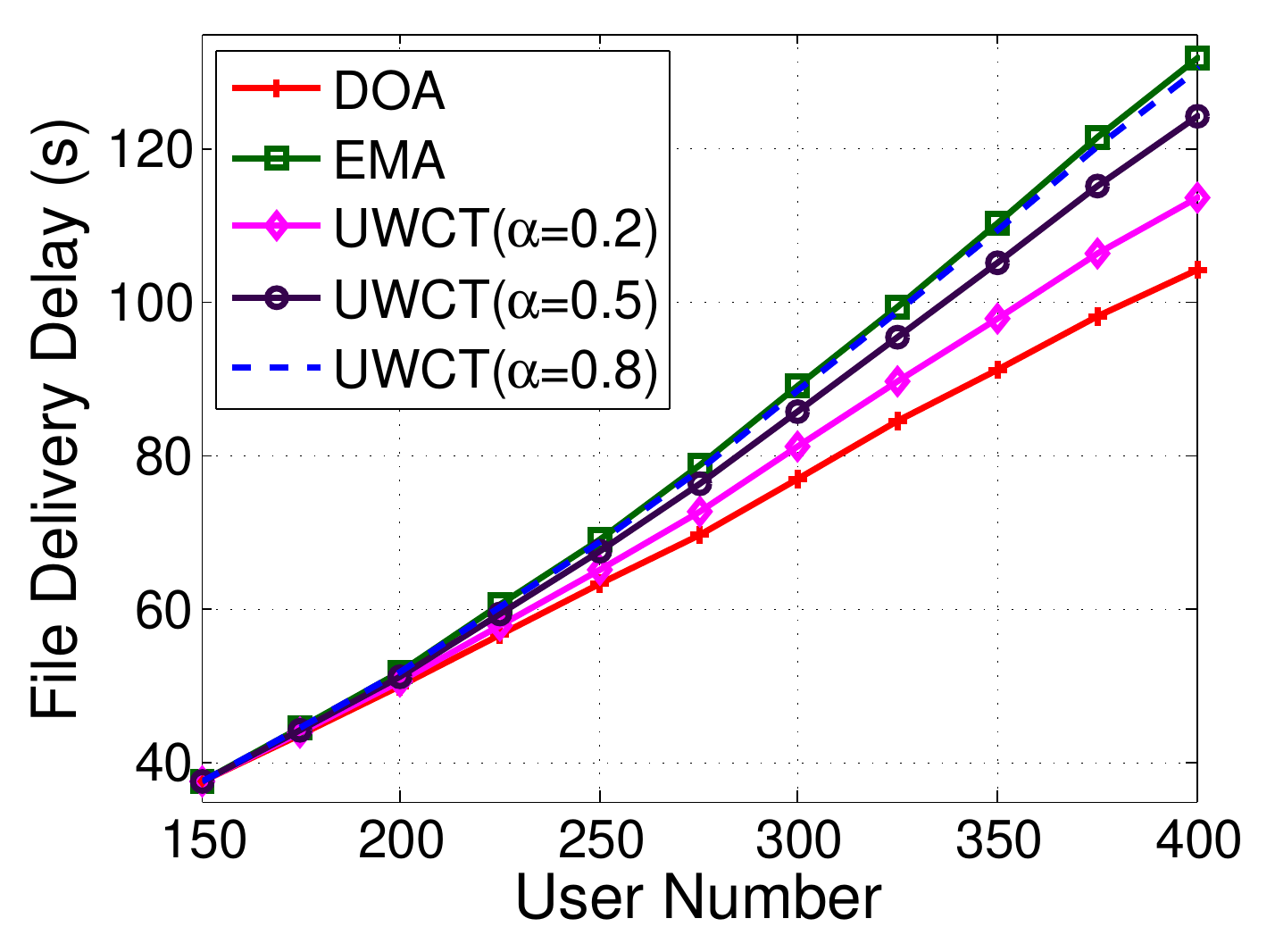}\label{UserDelayComparision} }                                        
\caption{(a) energy consumption and (b) end-to-end file delivery delay comparison under different numbers of user}

\label{UserDEComparision}
\end{figure}

\begin{figure}[ht]
\centering
\subfigure[]{\includegraphics[width=4.4cm]{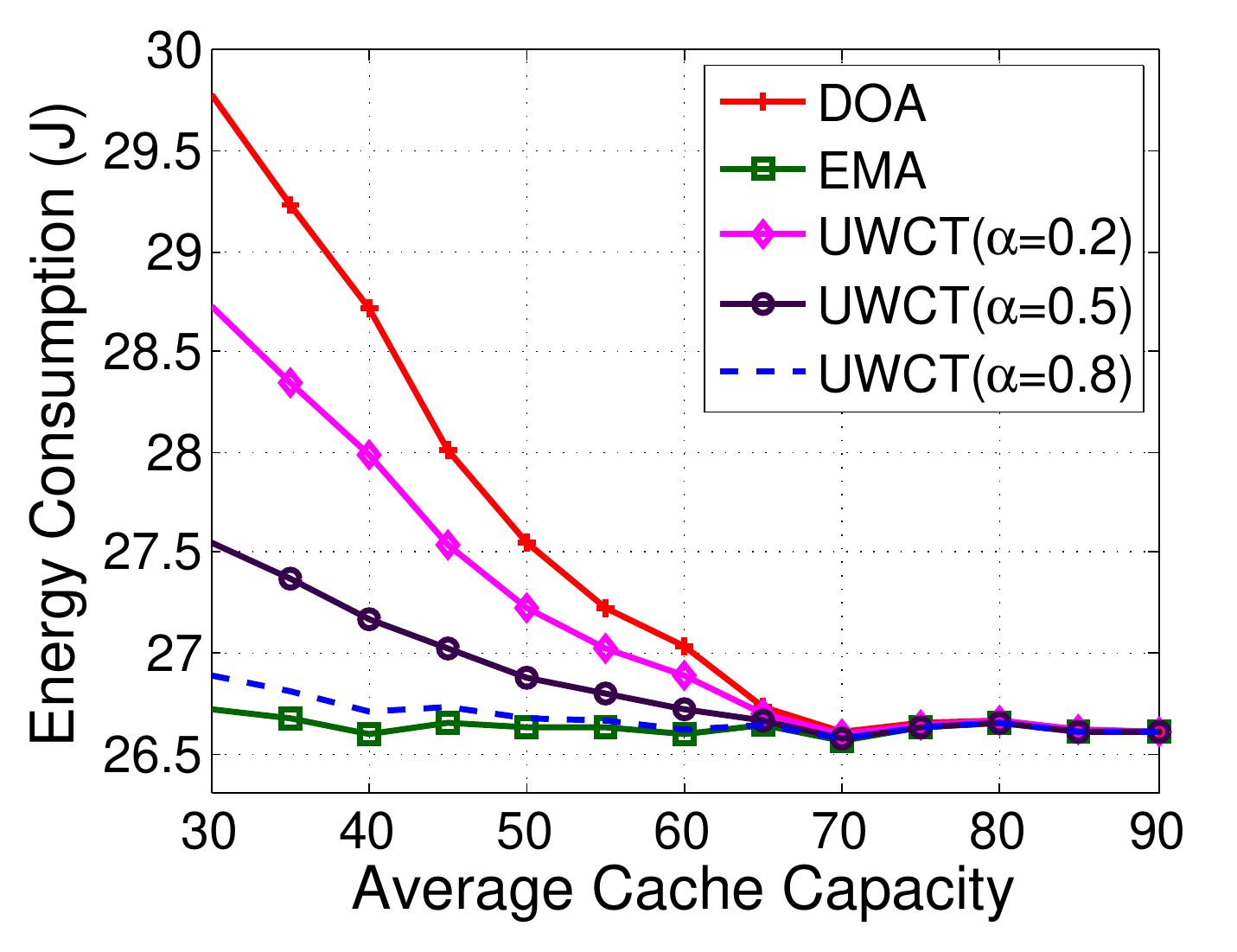}\label{CapacityEnergyComparision}}
\hspace{-2.3ex}
\subfigure[]{\includegraphics[width=4.4cm]{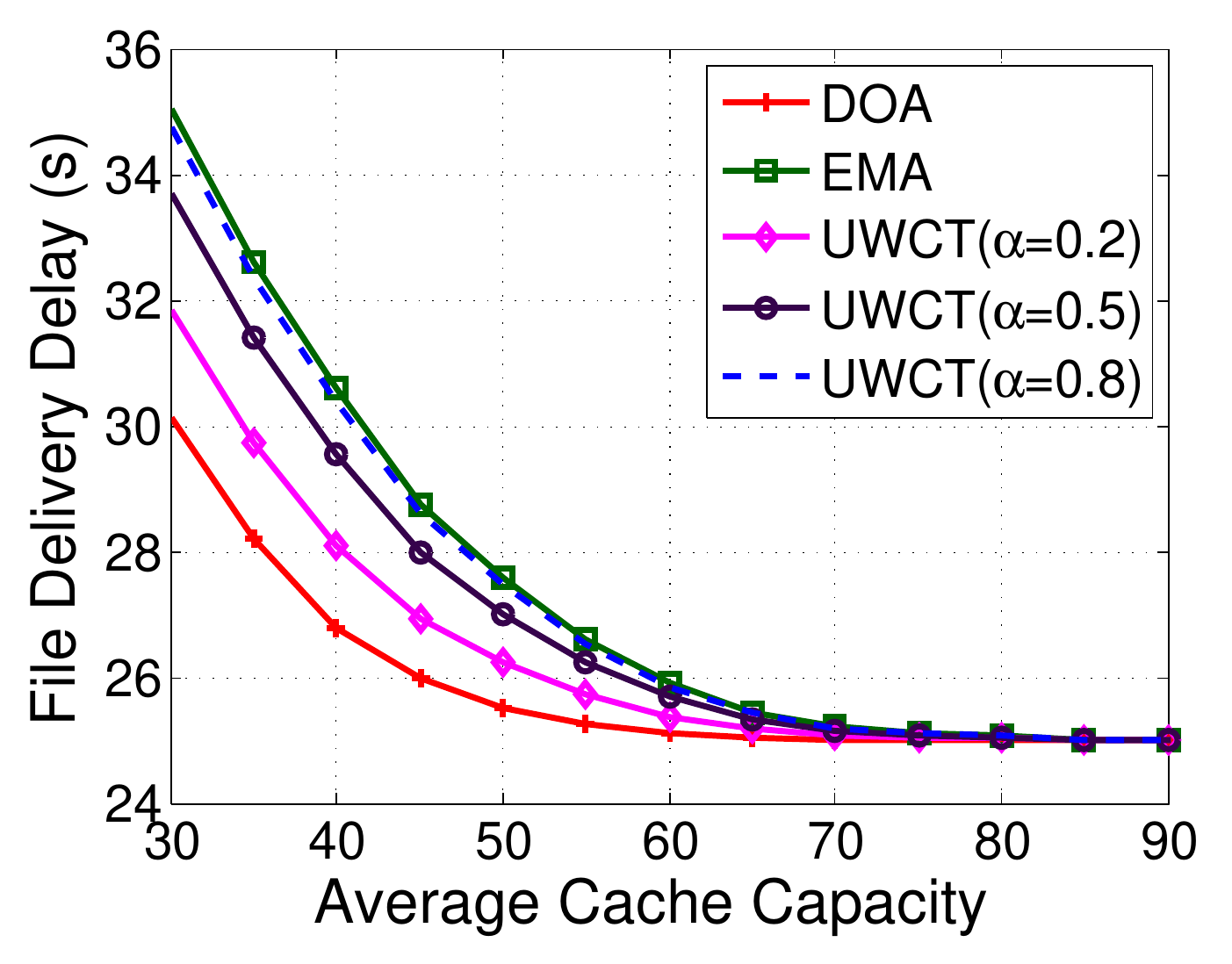}\label{CapacityDelayComparision}}                                         \caption{(a) energy consumption and (b) end-to-end file delivery delay comparison under different cache capacities}
\label{CapacityDEComparision}
\end{figure}

To demonstrate the advantages of the proposed energy-delay tradeoff strategy in cache-enabled DSCNs, we refer to the delay-optimal and energy-minimum algorithm under different numbers of  users and cache capacities.
\begin{itemize}
  \item \emph{Delay-Optimal Algorithm (DOA):} In order to minimize end-to-end file delivery delay, a user will only  associate with the SBS caching the requested file from its SBS neighbourhood \cite{3u}. This procedure will repeat until no more delay increases.
  \item \emph{Energy-Minimum Algorithm (EMA):} In the system model, each user is associated with the closest SBS, and the SBS chooses the most local popular contents in its cache\cite{5u}.
\end{itemize}

In Fig. \ref{UserDEComparision}, we compare the three algorithms in terms of energy consumption and end-to-end file delivery delay by varying the number of users. The average capacity is 15. In Fig. \ref{UserEnergyComparision}, as expected, EMA consumes least energy. The proposed UCWT algorithm consumes more energy than EMA, but always less than DOA. In Fig. \ref{UserDelayComparision}, as expected, DOA achieve minimum end-to-end file delivery delay. UCWT results in higher end-to-end file delivery delay than DOA, but always lower than EMA. This is due to the fact that EMA focuses on energy consumption minimization, which sacrifices end-to-end file delivery delay. It is opposite for DOA that focuses on optimizing delay. We can also see that UCWT can obtain a balance between energy consumption and end-to-end file delivery delay by adjusting tradeoff parameter $\alpha$. For example, when $\alpha$ is smaller, lower end-to-end file delivery delay is achieved while more energy is consumed.

In Fig. \ref{CapacityDEComparision}, we compare the three algorithms in terms of energy consumption and end-to-end file delivery delay by changing the average cache capacity. The number of user is 150. In Fig. \ref{CapacityEnergyComparision}, the energy consumption value obtained by EMA is almost constant over different average cache capacities. The reason is that users only associate with nearest SBSs without consideration of cached files at SBSs. In Fig. \ref{CapacityDelayComparision}, compared with EMA and UCWT, DOA  achieves  minimum end-to-end file delivery delay. From Fig. \ref{CapacityEnergyComparision} and Fig. \ref{CapacityDelayComparision}, when the average cache capacity is large enough, both minimum end-to-end file delivery delay and energy consumption are achieved by the three algorithms. This is due to the fact that all users can get required files from their local nearest SBSs.

\section{Conclusion}
In this paper, we study energy consumption and end-to-end file delivery delay tradeoff problem in cache-enabled DSCNs, where file caching, user association and power control are jointly considered. To solve the problem, firstly, a local popular file placement policy is proposed to maximize the caching hit probability at SBSs. With the proposed file placement policy, the tradeoff problem is further decomposed with Benders' decomposition method. Extension simulations show the proposed algorithms can obtain the desired energy-delay tradeoff under various scenarios.

In the future, we will extend our work to the mobility environments. Furthermore, machine learning based mechanisms will be considered to estimate the file popularity distribution at SBSs.

\appendices
\section{Proof of the proposition \ref{pro_master_problem2}}
Based on the definition in (\ref{set}), objective (\ref{master_problem}) can be equivalently transformed to
\begin{equation}\label{appen master problem}
\begin{split}
\min_{\eta,\bm{X}} & \quad\alpha\eta + (1-\alpha)\sum\limits_{i=1}^{U} \sum\limits_{j=1}^{B}  \sum\limits_{k=1}^{F}  \theta_{ik}d_{ij}^{k}x_{ij}, \\
   s.t.& \quad x_{ij}\in [0,1], \forall i\in U, \forall j\in B,\\
       & \quad \sum_{i=1}^{U}\sum_{j=1}^{B} x_{ij}-x_{ij}^{2}\leq 0,\\
       & \quad \text{remaining constraints is the same as in (\ref{master_problem})}.
\end{split}
\end{equation}

The Lagrangian function of (\ref{appen master problem}) with only one Lagrangian multiplier $\lambda\geq0$ (which leads to (\ref{master_problem2}))is
\begin{equation}\label{Lagrangian function}
\begin{split}
  \mathcal{L}(\eta,\bm{X},\lambda): & =\alpha\eta+(1-\alpha)\sum\limits_{i=1}^{U} \sum\limits_{j=1}^{B}  \sum\limits_{k=1}^{F}  \theta_{ik}d_{ij}^{k}x_{ij} \\
    & +\lambda\sum_{i=1}^{U}\sum_{j=1}^{B}(x_{ij}-x_{ij}^{2})
\end{split}
\end{equation}

The optimization problem (\ref{master_problem}) can be expressed by $\min\limits_{(\eta,\bm{X})}\max\limits_{\lambda\geq 0} \mathcal{L}(\eta,\bm{X},\lambda)$ (\ref{b}). According to the duality theory in \cite{dual theory}. So
\begin{align}
  \label{a}\sup_{\lambda} &\quad \phi(\lambda)=\sup_{\lambda} \min_{(\eta,\bm{X})} \mathcal{L}(\eta,\bm{X},\lambda),\\
  \label{b}               &\quad \quad \quad \leq \min_{\eta,\bm{X}}\max_{\lambda}\mathcal{L}(\eta,\bm{X},\lambda),\\
                 &\quad \quad \quad =\min (\ref{master_problem}),\nonumber
\end{align}

where $\phi(\lambda)=\min\limits_{(\eta,\bm{X})} \mathcal{L}(\eta,\bm{X},\lambda)$ and $\phi(\lambda)$ is function over $\lambda$. When $\sum_{i=1}^{U}\sum_{j=1}^{B}x_{ij}-\sum_{i=1}^{U}\sum_{j=1}^{B}x_{ij}^{2}\geq 0 $ for $x_{ij} \in [0,1],\forall i,j$, $\mathcal{L}(\eta,\bm{x}, \lambda)$ is monotonically increasing over $\lambda$ for $\forall \bm{X} \in \bm{A}, \forall \eta$, then $\phi(\lambda)$ is increasing in $\lambda$ and bounded by the optimal value of (\ref{master_problem}). Let the optimal solution for $(\ref{a})$ is denoted by $\lambda^{*},\eta^{*}$ and $\bm{X}^ {*}$, where $\lambda^{*}\in (0,+\infty)$. Then, the following two cases should be analyzed for the optimal solution of (\ref{a}).

\begin{itemize}
  \item The first case is when $\sum_{i=1}^{U}\sum_{j=1}^{B}x_{ij}^{(*)}-\sum_{i=1}^{U}\sum_{j=1}^{B}(x_{ij}^{(*)})^{2}=0$. At this time, $\eta^{*}$ and $X^ {*}$ are still feasible to (\ref{master_problem}). Then, when $\eta=\eta^{*}$ and $\bm{X}=\bm{X}^ {*}$, we have
\begin{equation}\label{d}
\begin{split}
\phi(\lambda^{*})&=
 \mathcal{L}(\eta^{*},\bm{X}^{*},\lambda^{*})\\
 &=\alpha\eta^{*}
  +(1-\alpha)\sum\limits_{i=1}^{U} \sum\limits_{j=1}^{B}  \sum\limits_{k=1}^{F}  \theta_{ik}d_{ij}^{k}x_{ij}^{*} \\
  &\geq \min(\ref{master_problem}),\\
\end{split}
\end{equation}

Look back at (\ref{a}) and (\ref{d}), such following equation holds:
\begin{align}
   \sup_{\lambda} \min_{(\eta,\bm{X})} \mathcal{L}(\eta,\bm{X},\lambda)=\min_{\eta,\bm{X}}\max_{\lambda}\mathcal{L}(\eta,\bm{X},\lambda),\nonumber
\end{align}
when $\sum_{i=1}^{U}\sum_{j=1}^{B}x_{ij}-\sum_{i=1}^{U}\sum_{j=1}^{B}(x_{ij})^{2}=0$.

As $\phi(\lambda)$ is monotonically increasing function over $\lambda$. Then
\begin{equation}\label{e}
  \phi(\lambda)=\min(\ref{master_problem}), \forall \lambda\geq \lambda^{*}
\end{equation}

Namely, (\ref{master_problem2}) and (\ref{master_problem}) share the same optimal solutions and value, where $\sum_{i=1}^{U}\sum_{j=1}^{B}x_{ij}-\sum_{i=1}^{U}\sum_{j=1}^{B}(x_{ij})^{2}=0$. Thus, proposition (\ref{pro_master_problem2}) holds for $\sum_{i=1}^{U}\sum_{j=1}^{B}x_{ij}-\sum_{i=1}^{U}\sum_{j=1}^{B}(x_{ij})^{2}=0$.
  \item The second case is that we assume $\sum_{i=1}^{U}\sum_{j=1}^{B}x_{ij}-\sum_{i=1}^{U}\sum_{j=1}^{B}(x_{ij})^{2}>0$, $\lambda > 0$  for optimizing (\ref{a}). Due the monotonicity of function $\phi(\lambda)$ over $\lambda$, $\max\limits_{\lambda\geq 0}\phi(\lambda)$ tends to $+\infty$ with $\lambda$$\rightarrow$$+\infty$. Such result contradicts the conclusion that (\ref{a}) is less than the optimal value of $(\ref{master_problem})$ in expression (\ref{b}). Thus, there exists $\hat{x}_{ij}$ satisfying $\sum_{i=1}^{U}\sum_{j=1}^{B}\hat{x}_{ij}-\sum_{i=1}^{U}\sum_{j=1}^{B}(\hat{x}_{ij})^{2}=0$ with $\lambda$$\rightarrow$$+\infty$.
\end{itemize}

Based on the above analysis, we can conclude that when an appropriate value is chosen for $\lambda$, the problem (\ref{master_problem}) is equivalent to the problem (\ref{master_problem2}) in the sense that they share the same optimal value as well as optimal solution.

\ifCLASSOPTIONcaptionsoff
  \newpage
\fi



%

\end{document}